\documentclass[11pt]{article}

\usepackage[a4paper,margin=2cm]{geometry} 
\usepackage{setspace}
\usepackage{xspace}
\usepackage[dvipsnames]{xcolor}

\usepackage{amsmath,amssymb,amsthm,mathtools}

\newtheorem{algorithm}{Algorithm}
\newtheorem{theorem}{Theorem}
\newtheorem{lemma}{Lemma}
\newtheorem*{theorem*}{Theorem}

\usepackage{enumitem}

\usepackage{hyperref}
\hypersetup{
  colorlinks=true,
  linkcolor=black,
  urlcolor=blue
}
\usepackage{cleveref}

\usepackage{tcolorbox}

\usepackage{caption}
\newcommand{\N}{\ensuremath{\mathbb{N}}\xspace}
\newcommand{\R}{\ensuremath{\mathbb{R}}\xspace}
\newcommand{\ER}{\ensuremath{\exists\mathbb{R}}\xspace}
\newcommand{\PR}{\ensuremath{\text{P}\mathbb{R}}\xspace}
\newcommand{\NP}{\ensuremath{\text{NP}}\xspace}
\newcommand{\PSPACE}{\ensuremath{\text{PSPACE}}\xspace}

\newcommand{\inv}{\ensuremath{\texttt{inv}}\xspace}

\renewcommand{\L}{\ensuremath{\mathcal{L}}\xspace}

\newcommand{\wordRAM}{word RAM\xspace}
\newcommand{\WordRAM}{Word RAM\xspace}
\newcommand{\realRAM}{real RAM\xspace}
\newcommand{\RealRAM}{Real RAM\xspace}
\newcommand{\PosSLP}{\textsc{PosSLP}\xspace}
\newcommand{\stretchability}{\textsc{Stretchability}\xspace}
\newcommand{\ordertype}{\textsc{Order Type}\xspace}
\newcommand{\etrinv}{\textsc{ETRINV}\xspace}
\newcommand{\etr}{\textsc{ETR}\xspace}
\newcommand{\satisfiability}{\textsc{Satisfiability}\xspace}
\newcommand{\artgalleryProblem}{\textsc{Art Gallery Problem}\xspace}
\newcommand{\partialOrderTypeProblem}{\textsc{Partial Order Type Problem}\xspace}
\newcommand{\OrderTypeProblem}{\textsc{Order Type Problem}\xspace}

\newcommand{\Vis}{\operatorname{Vis}}

\title{Beyond Bits:\\  An Introduction to Computation over the Reals}
\author{Till Miltzow}
\date{}

\begin{document}

\maketitle

We aim to introduce, in a lightweight and accessible way, how to think about computation over the real numbers and why such models are useful.
The material is intended for multiple audiences: instructors who wish to include real computation in an algorithms course, their students, and PhD students encountering the subject for the first time.

We cover a range of basic concepts and present several complementary techniques and perspectives.
Rather than aiming for completeness, we deliberately cherry-pick a small number of results that can be proved in a classroom setting, while still conveying some of the technical core ideas that recur throughout the field.
As a consequence, the exposition is occasionally informal.
This is a conscious choice: our goal is not full technical precision at every point, but to communicate the underlying ideas and to illustrate how the concepts are used in research practice.

There are only a small number of modern introductory texts on computation over the reals.
Most prominently, Ji{\v{r}}{\'\i} Matou{\v{s}}ek’s lecture notes~\cite{Matousek2014ETR} provide a beautifully streamlined proof of the \ER-completeness of stretchability, one of the central \ER-complete problems.
Along the way, the reader is introduced to the complexity class itself through a single, coherent narrative.
A second reference is the recent compendium by Schaefer, Cardinal, and Miltzow~\cite{SchaeferCardinalMiltzow2024Compendium}, whose goal is to provide a comprehensive catalogue of known \ER-complete problems.
As a reference work, it is necessarily dense and does not include proofs or exercises.

The book by Blum, Cucker, Shub, and Smale~\cite{BlumCuckerShubSmale1998} is a foundational reference for computation over the real numbers and played a crucial role in establishing the subject as a coherent area of study.
It introduces a real computation model and develops a rich theory of complexity over continuous domains.
At the same time, the computational model adopted there differs in several respects from formulations that are most commonly used today.
In particular, the definition of real Turing machines is technically involved, and the model allows arbitrary access to real constants, which complicates comparisons with classical discrete complexity theory.
Modern treatments typically emphasize binary input, real-valued witnesses, and a restricted use of constants, leading to models that align more closely with contemporary notions of computation and complexity.
Our approach follows this more recent perspective, while still viewing the framework by Blum, Shub and Smale as an important historical and conceptual precursor.

\newpage
\tableofcontents

\newpage
\section{Models of Computation}

\subsection{Motivation}
An algorithm is a description of how we can do computations on a computer.
Often, those algorithms are given using free text, as in \Cref{alg:freetextLinearSearch}.

\begin{algorithm}[Linear Search]
\label{alg:freetextLinearSearch}
Given a finite list of numbers $L = (x_1, x_2, \dots, x_n)$ and a number $t$, we inspect the elements of the list from left to right.
Once we find an $x_i$ such that $x_i = t$, then the algorithm returns the index $i$.
If no such index exists, the algorithm reports that $t$ does not occur in the list.
\end{algorithm}

Students and teachers might find this style not precise enough and prefer a description using pseudocode, as in \Cref{alg:pseudocodeLinearSearch}.

\begin{algorithm}[Linear Search]
\label{alg:pseudocodeLinearSearch}

\ 

\noindent \textbf{Input:} A finite list of numbers $L = (x_1, x_2, \dots, x_n)$ and a number $t$ (the target).

\noindent \textbf{Output:} The position of $t$ in the list, or the statement ``$t$ is not in the list''.

\noindent \textbf{Procedure:}
\begin{enumerate}[noitemsep,topsep=0pt]
  \item Start with the first element of the list.
  \item Compare the current element with the target $t$.
  \item If they are equal, stop and report the position of this element.
  \item Otherwise, move to the next element in the list.
  \item Repeat steps 2--4 until either the target is found or the end of the list is reached.
  \item If the end of the list is reached without finding $t$, report that $t$ is not in the list.
\end{enumerate}
\end{algorithm}

This second version adds at points more detail and helps the reader to clarify ambiguity. 
It serves two purposes. 

First, it is a useful step towards implementing this algorithm on a computer.
Second, this is often used in introductory algorithms classes for mathematical analysis.
We can show that the algorithm always terminates, always returns the correct solution, takes a total of $X$ steps, and uses at most $Y$ words of working memory.
But if we take a step back, we normally carry out mathematical analysis on mathematically well-defined objects
such as sets, numbers, functions, relations, vector spaces, and many more.
In what sense is either of the two descriptions of an algorithm a well-defined mathematical object?
In the following, we define an algorithm as a well-defined mathematical object and explain why the descriptions above are fine and desirable to use, even if we know how to define algorithms more rigorously.

\subsection{\WordRAM}
The \wordRAM model is meant to capture how physical computers actually handle data, while still being simple enough to reason about mathematically.
The abbreviation RAM stands for Random Access Memory, as it can access all its memory without restrictions.
The central idea of the \wordRAM is that the machine works with \emph{words} of fixed size, rather than with individual bits.
A word is large enough to store any memory address or any integer of reasonable size, typically on the order of $\Theta(\log n)$ bits when the input size is $n$.

\begin{center}
    \includegraphics{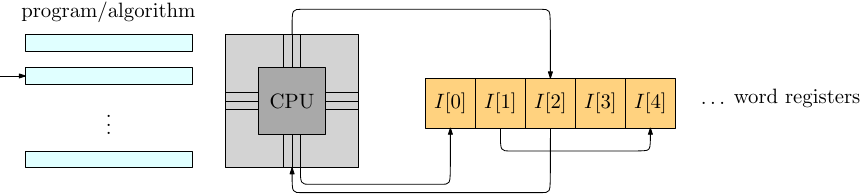}
    \captionof{figure}{An algorithm on the \wordRAM is a finite list of instructions. Input, output, and all working memory take place in an infinite array of registers.}\label{fig:WordRAM}
\end{center}

Memory is viewed as a large array of such words, and each word can be accessed directly using its address.
This reflects the fact that modern computers allow random access to memory: accessing the contents of a specific memory cell takes constant time, independent of where that cell is located.
The model therefore abstracts away from physical details such as caches or memory hierarchies, while keeping the essential feature of direct addressing.

The machine can perform basic operations on words in constant time.
These include arithmetic operations (such as addition and subtraction), comparisons, bitwise operations, and reading from or writing to memory.
The keypoint is that there is a fixed number of allowed primitive operations.

\begin{center}
\begin{tabular}{|l|p{0.68\linewidth}|}
\hline
\textbf{Operation type} & \textbf{Standard operations (unit cost)} \\
\hline
Assign Constants & $M[i] \gets j$, with $i,j$ being constants. \\
\hline
Copy & $M[i] \gets M[j]$, with $i,j$ being constants  \\
\hline
Copy via Indirect Access &  $M[i] \gets M[M[j]]$, with $i,j$ being constants \\
& This copies the content of the cell that $M[j]$ points to into $M[i]$.\\

\hline
Arithmetic on words & $M[i] = M[j] \circ M[k]$, with $i,j,k$ being constants and $\circ \in \{+,-,\times, / \}$  \\
& Can be combined with indirect access,i.e., $M[M[1]] = M[2] - M[3]$.\\
\hline
Bitwise operations & AND, OR, XOR, NOT on words. (optional) \\
\hline
Shifts & Left/right shift by a constant (optional). \\
\hline
Comparisons & $=$, $\neq$, $<$, $\le$, $>$, $\ge$ between words. \\
\hline
Control flow & Conditional branch (based on a comparison result) and unconditional jump (go to line $\ell$). 
 If $M[1] > M[2]$ GoTo line $\ell$.  \\
\hline
\end{tabular}
\captionof{table}{A typical set of standard primitive operations in the \wordRAM model.}\label{tab:wordram-ops}
\end{center}

One way to think of a program on the \wordRAM is as the machine code that is produced by a compiler on a physical computer.
It has only very primitive operations, but this is not how most people would program.
\Cref{fig:WordRAM} gives a visualization of the \wordRAM, the set of instructions, the central processing unit performing the instructions and the memory cells.
For ilustrative purposes, we implement the previous example of linear search being on the \wordRAM, see \Cref{alg:WordRAM-LinearSearch}.

\begin{algorithm}[LinearSearch on WordRAM]
\label{alg:WordRAM-LinearSearch}
We assume a \wordRAM with word size $w = \Theta(\log n)$.
Memory is an array $M[\cdot]$ of words.
The input is stored as follows:
$M[0] = n$, $M[1] = t$, and $M[3], \ldots, M[n+2]$ contain the list of elements $x_1,\ldots,x_n$.

\medskip

\noindent
\begin{tabular*}{\linewidth}{@{}l@{\extracolsep{\fill}}r@{}}
\textbf{1:} $M[0] \gets M[0] + 2$
  & adjust the length to match the final index \\

\textbf{2:} $M[2] \gets 3$
  & initialize the index variable \\

\textbf{3:} $\text{if } M[M[2]] = M[1] \text{ then return } M[2]$
  & found the number, return the index \\

\textbf{4:} $M[2] \gets M[2] + 1$
  & move to the next cell \\

\textbf{5:} $\text{if } M[2] > M[0] \text{ then return ``NO''}$
  & the number is not in the list \\

\textbf{6:} $\text{go to line 3}$
  & repeat the search loop
\end{tabular*}

\end{algorithm}

\Cref{tab:wordram-ops} gives a rough overview of possible allowed operations.
Even in the algorithm described above, we did not very strictly adhere to the set of allowed operations.

\paragraph{Mathematical Analysis.}
For mathematical analysis, the very features that make a \wordRAM program unpleasant to read becomes its strength.

Once an algorithm is written as a \wordRAM program, every operation has a precise and agreed-upon meaning and cost. There is no ambiguity about what counts as a step, what kind of memory access is allowed, or how large the numbers involved may be. This matters because analysis is ultimately about counting operations, not about intentions. A \wordRAM program fixes a concrete execution model on which such counting makes sense.

The explicit nature of the \wordRAM formulation forces us to be precise. 
This makes assumptions about random access, word size, and arithmetic power visible. 
As a result, claims like “this runs in linear time” or “this uses constant extra space” can be checked mechanically against the model rather than accepted informally.

\paragraph{Independent of Physical Computers.}
The \wordRAM also draws a clean boundary between algorithmic ideas and hardware assumptions. By idealizing word-level operations and constant-time memory access, it captures what is essential about modern machines while ignoring ``irrelevant'' engineering details. 
This makes it possible to reason mathematically about algorithms that rely on bit tricks, hashing, or packed data structures, without having to commit to a specific processor architecture.

\paragraph{Limited Communication Value.}
A \wordRAM program is excellent for analysis, but rather poor as a communication device.

The main reason is that it exposes details that are irrelevant for understanding the idea of an algorithm, while hiding the structure that humans actually reason about. 
A \wordRAM program forces the reader to think in terms of registers, memory locations, and jumps. 
As a result, the high-level logic—what is being searched for, why a loop terminates, what invariant is being maintained—has to be reconstructed mentally from low-level operations. This is cognitively expensive and error-prone, even for experienced programmers.

\WordRAM descriptions also scale poorly. 
As algorithms grow more complex, the number of instructions increases rapidly, and the global structure of the algorithm becomes fragmented across jumps and register updates. Properties that are obvious in higher-level descriptions—such as “repeat until the list ends” or “maintain the smallest value seen so far”—are no longer explicit but must be inferred.

\paragraph{Implicitness Assumption.}
In theoretical computer science, algorithms will typically be described in words and sometimes using high-level pseudocode.
It is implicitly assumed that any such description can, if necessary, be translated into an equivalent program in the \wordRAM model. 
We will not usually carry out this translation explicitly, since doing so adds little insight and obscures the underlying idea of the algorithm. 
Nevertheless the \wordRAM model therefore serves as a reference point for analysis rather than as a language for presentation.
And whenever we do a mathematical analysis of an algorithm, we do the analysis on the implicit \wordRAM program.

\subsection{Turing Machines}
Talking about a definition of computation without mentioning Turing machines is like cake without coffee: something is missing.
The main value of Turing machines, compared to the \wordRAM, is their mathematical simplicity.
They originated from the following thought experiment.

Imagine a very patient but extremely simple worker sitting in front of an endless strip of paper, divided into squares.
Each square can hold a single symbol. The worker can look at exactly one square at a time, change the symbol written there, and then move one step to the left or to the right.
(Immediately forgetting what they saw on the previous square.)
The worker also has a small instruction sheet—a finite list of rules—that tells them what to do depending on the symbol they see and on which line of the instruction sheet they are currently following.
In other words, the “memory” of the worker is just their place on this finite list of rules.
That is all the worker can do. There is no overview of the tape, no jumping to arbitrary positions, no arithmetic beyond what can be encoded in those tiny steps. Everything happens locally, one square at a time. 
Computation, in this model, is the process of repeatedly applying these simple rules, potentially for a very long time.
Instead of giving a mathematically formal definition, let's have a short discussion of such a machine.

\begin{center}
    \includegraphics[page = 3]{figures/RegiserMachine.pdf}
    \captionof{figure}{Turing Machines originated from a simple thought experiment about a simple worker.}\label{fig:Turing-Machine}
\end{center}

\paragraph{Mathematically Simplistic and Formal.}
Although this picture is intentionally informal, the model can be made fully mathematically rigorous. The worker, the tape, and the instruction sheet can all be described using finite sets, functions, and precise transition rules, leaving no ambiguity about how the computation proceeds. At the same time, the model is mathematically simplistic: it relies on only a handful of basic operations and avoids any appeal to advanced mathematics or physical assumptions. This combination is deliberate. By keeping the model as simple as possible while still capturing all computable procedures, the Turing machine provides a clean and robust foundation for rigorous reasoning about what computation can and cannot do.

\paragraph{Equivalent to \wordRAM.}
A Turing machine and the \wordRAM model differ greatly in how computation looks, but they are closely related in power. Any algorithm that can be executed on a \wordRAM can be simulated by a Turing machine with at most a polynomial slowdown, and conversely, any Turing machine running can be simulated on a \wordRAM with only constant multiplicative overhead. 
(The simulation of a \wordRAM program on a Turing machine is tedious and requires a bunch of clever tricks.)
As a consequence, when we classify algorithms by polynomial time versus superpolynomial time, the choice between Turing machines and \wordRAM{}s does not matter: they define the same notion of efficient computation, up to polynomial factors.

\paragraph{Algorithm Design on Turing Machines.}
At first it might feel difficult to design algorithms on a single-tape Turing machine, because at any time the worker can store only finitely much information.
A basic trick is to use explicit start and end markers and to program the head to sweep back and forth between them, ``remembering'' information by writing it onto the tape and picking it up on later passes.
For algorithm design it is often more convenient to work with a \emph{multitape} Turing machine: besides the input tape we have a constant number of work tapes, each with its own head, so we can keep, say, a binary counter on one tape, a marked current position on another, and scratch space on a third.
Interestingly, multitape machines can be simulated with the standard single-tape model.
The classical simulation encodes all tapes on one tape separated by delimiters, and marks the simulated head position on each encoded track.
To simulate one multitape step, the single-tape machine scans to find the marked symbols, determines the next transition, and then makes another pass to update the affected symbols and move the marks.
This causes only polynomial overhead.
To formulate all those tricks and to show that the Turing machine can indeed simulate the \wordRAM with only polynomial overhead might take several lectures to fully explain, so we skip it here.

\subsection{Real Models of Computation}

Real models of computation are abstract models in which real numbers are treated as basic objects of computation, rather than being encoded as finite bit strings. In these models, operations such as addition, subtraction, comparison, and multiplication or division on real numbers are assumed to take constant time and to be exact. 
Before, we give more details on how this is defined, let's look at some reasons why one might or might not want to consider those models of computation.

\paragraph{Hide Numerical Issues.}
One motivation for real models of computation is that they deliberately abstract away from numerical representation issues such as rounding, precision loss, and floating-point errors. 
In many algorithmic settings—especially in geometry or optimization—these issues are not central to the combinatorial or structural difficulty of the problem, but can dominate the analysis when working in bit-level models. 
By assuming exact real arithmetic, real models allow us to study the logical structure and asymptotic behavior of algorithms without being distracted by numerical concerns. 
This can clarify which difficulties are truly algorithmic and which arise only from finite-precision representations.
(Yet those issues might bite you again later.)
This is the main reason that the \realRAM is used in Computational Geometry.

\paragraph{Unrealistic Assumption.}
A common criticism of real models of computation is that they rely on assumptions that cannot be realized on physical machines. 
Real numbers generally require infinite precision, and exact arithmetic operations on them cannot be performed in constant time on actual hardware. 
As a result, algorithms that appear efficient in a real model may hide substantial computational costs when translated to bit-level implementations. 
For this reason, results proved in real models must be interpreted with care.

In extreme cases, this idealization allows one to effectively ``cheat'' by encoding an enormous amount of information into a single real number and manipulating it using constant-time operations, potentially collapsing computations that would take exponential time in more realistic models into polynomial time. 
We will show an example on how this can be done.

\paragraph{Naturalness.}
In some settings, real models of computation arise not because one is explicitly interested in real numbers, but because they enter the problem formulation in a natural and unavoidable way. 
Even when the underlying question is combinatorial or logical in nature, continuous parameters, coordinates, or thresholds may appear as part of a clean mathematical description. 
In such cases, attempting to eliminate real numbers entirely can lead to artificial encodings or unnecessarily complicated definitions. 
Introducing a real model of computation then becomes a methodological step: it provides a formal framework in which these naturally occurring real quantities can be handled directly and consistently. The model does not assert that the problem is “about” real numbers, but that allowing them simplifies the expression and analysis of the underlying computational phenomenon.
One example is convex optimization that often hides their model of computation, which necessarily needs to involve real numbers the way that their algorithms and their analysis is formulated.
We will give more concrete examples later.

\begin{center}
    \includegraphics[page = 2]{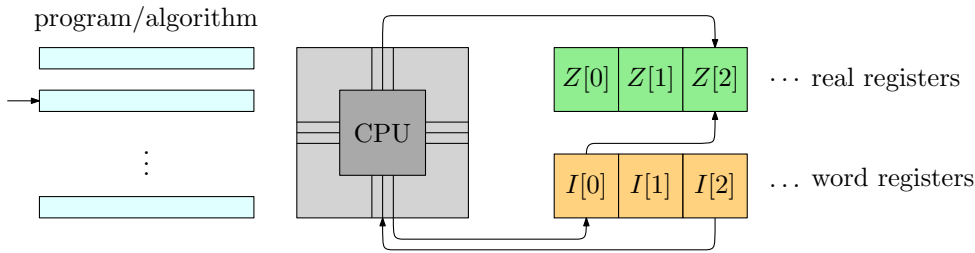}
    \captionof{figure}{The \realRAM has two infinite arrays of registers.
    One that can operate on words and one that is allowed to operate on real numbers.}\label{fig:realRAM}
\end{center}

\paragraph{\RealRAM and Real Turing Machine.}
The two models of computation that we had seen earlier can be extended to the \realRAM and the real Turing machine respectively.
We give now a description on how one might do this.

A \realRAM can be viewed as a direct extension of the \wordRAM model by enriching it with an additional set of registers that store real numbers instead of machine words. 
Alongside the usual word registers, which hold integers, addresses, and control information, the machine is equipped with real registers that can contain arbitrary real values as input. 
The instruction set is extended accordingly: in addition to word-level operations, the machine may perform constant-time arithmetic operations and comparisons on the contents of real registers.
These operations are $+,-,\times,\div$, and a sign test of the form $R>0$.
The last operation returns yes if the content of register $R$ is greater than zero, and no otherwise.
Note that we cannot simply use arbitrary real constants such as $\pi$ or take square roots.
In this way, the \realRAM preserves the structural features of the \wordRAM---random access and a clear cost model---while allowing algorithms to manipulate real numbers directly as atomic objects.
It is important that the set of real operations do not give any direct access to the bit representation of the real numbers, this includes the rounding operation.

A real Turing machine can be described as a natural extension of the classical Turing machine by adding an extra tape dedicated to real numbers. 
In addition to the usual tape, which stores symbols from a finite alphabet and is used for control and bookkeeping, the machine is equipped with a separate real tape whose cells contain real values. The transition rules are extended so that the machine may read, write, and compare real numbers on this tape, and perform basic real arithmetic in a single step. As with the standard Turing machine, the control remains finite and all actions are local. 
Again it is important that the access to the real number is limited to comparison and real arithmetic.

\subsection{Cheating with the \RealRAM.}
There are many ways that one might give access to the number representation to the \realRAM. 
One such way is to allow rounding. 
Specifically given $x\in \R$ we return the unique integer $i$ such that $i\leq x < i+1$.
One way to use this power is to solve the \textsc{Factoring Problem}.
That is: find a factor of a natural number $n$. 
We show the following theorem due to Shamir~\cite{Shamir1977}.

\begin{theorem}
    \label{thm:FastFactoring}
    We can solve the \textsc{Factoring Problem} in polynomial time on a \realRAM that has a rounding operation.
\end{theorem}
\begin{proof}
    The idea is that $\gcd(n,k!) > 1$ if and only if $n$ has a nontrivial factor at most $k$.
    We need three ingredients to turn this into an algorithm.
    First, we need a fast way to compute $k$ factorial.
    Second, we need a way to compute the greatest common divisor of two numbers quickly using rounding.
    Third, we do binary search on the value of $k$.
    We explain the first two steps and then outline the overall procedure.
    
   \paragraph{Computing Factorials Fast.}
    We prove the following lemma.
    \begin{lemma}\label{lem:factorial-wordram}
    Let $n$ be an integer whose binary representation has length $b=O(\log n)$. Then $n!$ can be computed in $O(b^2)$ steps on a \realRAM with the rounding operation.
    \end{lemma}

    As preparation, and to become familiar with the techniques, we first prove the following lemma.
    \begin{lemma}
      \label{lem:fastExponents}
      Let $k,n$ be integers and let $n$ have a binary representation of length $b=O(\log n)$. Then $k^n$ can be computed in $O(b)$ steps on a \realRAM.
    \end{lemma}
    \begin{proof}
      We do not need rounding, and we assume that $k$ is already stored in a real register.
      If $n = 2m$ is even, then
      \[k^n = k^{2m} = \left( k^{m} \right )^2.\]
      If $n = 2m + 1$ is odd, then
      \[k^n = k^{2m + 1 } = \left( k^{m} \right )^2 \cdot k.\]
      Using this recursion (i.e., exponentiation by squaring), we can compute $k^n$ in $O(b)$ steps, with each step taking constant time.
    \end{proof}

    We are now ready for the proof of \Cref{lem:factorial-wordram}.
    First note that if $n$ is odd, we can use the recursion $n! = n \times (n-1)!$.
    This occurs at most every second step and takes constant time.
    Now assume that $n = 2m$.
    Then
    \[\binom{2m}{m} = \frac{(2m)!}{(m!)^2} \leq 2^{2m}.\]
    Hence,
    \[(2m)! = \binom{2m}{m} (m!)^2.\]
    Thus, we can recurse provided that we can compute $\binom{2m}{m}$
    in $O(b)$ steps. We show this in the remainder of this paragraph.
    First, recall the binomial identity
    \[(x+y)^k = \sum_{i=0}^{k} \binom{k}{i} x^i y^{k-i}.\]
    Now, for $x=2^l$, $y=1$, and $k=2m$, we obtain
    \[(2^l+1)^{2m} = \sum_{i=0}^{2m} \binom{2m}{i}\, 2^{l i}.\]
    Using repeated squaring, we can compute $(2^l+1)^{2m}$ efficiently, 
    see \Cref{lem:fastExponents}.
    Choose an integer block length $l$ with $l \ge 2m$ (for instance, $l=2m$ suffices).
    Then each term $\binom{2m}{i}2^{li}$ occupies (at most) the $l$ consecutive bits in positions $li,\ldots,l(i+1)-1$:
    indeed, $\binom{2m}{i}\le 2^{2m}\le 2^l$, so adding these shifted terms causes no carries between different $l$-bit blocks.
    Consequently, the binary expansion of $(2^l+1)^{2m}$ consists of the binomial coefficients $\binom{2m}{0},\binom{2m}{1},\ldots,\binom{2m}{2m}$ written one after another in blocks of length $l$.
    As this is perhaps a bit difficult to see, we look at the example
    \[(2^6 + 1)^6 = \binom{6}{0} 2^{0\cdot 6} + \binom{6}{1} 2^{1 \cdot 6} + \binom{6}{2} 2^{2 \cdot 6} + \binom{6}{3} 2^{3 \cdot 6} + \binom{6}{4} 2^{4 \cdot 6} + \binom{6}{5} 2^{5 \cdot 6} + \binom{6}{6} 2^{6 \cdot 6}.\]
    Now, when we calculate $(2^6 + 1)^6$ in binary, we get 
    \[(2^6 + 1)^6 =  \texttt{ 000001 \ 000110 \  001111 \  010100 \  001111 \  000110 \  000001}.
 \]
  Now we compute the binomial coefficients as well, and we see that they agree with the binary digits of $(2^6 + 1)^6$:
    \[
    \begin{array}{rcl}
    \binom{6}{0} &=& 1  \;=\; \texttt{000001} \\
    \binom{6}{1} &=& 6  \;=\; \texttt{000110} \\
    \binom{6}{2} &=& 15 \;=\; \texttt{001111} \\
    \binom{6}{3} &=& 20 \;=\; \texttt{010100} \\
    \binom{6}{4} &=& 15 \;=\; \texttt{001111} \\
    \binom{6}{5} &=& 6  \;=\; \texttt{000110} \\
    \binom{6}{6} &=& 1  \;=\; \texttt{000001}
    \end{array}
    \]

    We now show how to extract a contiguous block of binary digits from a large integer.
    \begin{lemma}
      Let $R$ be an integer stored in a real register, 
      with binary representation
      $b_k b_{k-1} \ldots b_{0}$
      and let $0\le a < b \leq k$ be integers.
      Assume we do computations on a \realRAM with the rounding operation.
      Then we can compute the number $S$ with binary representation 
      $b_{b-1} b_{b-2} \ldots b_{a+1} b_{a}$ in constant time.
    \end{lemma}
    \begin{proof}
      Let
      \[
        R_1 \;=\; \left\lfloor \frac{R}{2^a}\right\rfloor.
      \]
      This discards the $a$ least significant bits of $R$.

      Next, we discard all bits of weight at least $2^{b-a}$ by reducing $R_1$ modulo $2^{b-a}$.
      Using the identity $x \bmod M = x - \lfloor x/M\rfloor M$, we define
      \[
        S \;=\; R_1 - \left\lfloor \frac{R_1}{2^{b-a}}\right\rfloor \,2^{b-a},
      \]
      which is exactly the required block of bits of $R$.
    \end{proof}

    We apply this lemma with $R  = (2^l+1)^{2m}$ and choose the block boundaries $a=lm$ and $b=l(m+1)$.
    This extracts exactly the $l$-bit block encoding $\binom{2m}{m}$, as required, in constant time.

\paragraph{Computing the Greatest Common Divisor.}
    We will show the following lemma.
    \begin{lemma}
    \label{lem:Euclidalgorithm}
    Let $a,b$ be natural numbers.
    There is an algorithm that computes the greatest common divisor $\gcd(a,b)$ in $O(\min (\log a , \log b ))$ time steps on a \realRAM with the rounding operation.
    \end{lemma}

    The standard algorithm for computing the greatest common divisor is Euclid's algorithm.
    Here we instead use a simpler but slower algorithm that is easier to analyze. 
    For notational convenience, we assume throughout the following proof that $a\leq b$.

    \begin{proof}
    First, we show how to compute the remainder $r < a$ such that
    \(b = a \cdot m + r\) and replace $b$ by $r$.
    We have
    \[
      r = b - \lfloor b/a \rfloor \cdot a.
    \]
    Note that $\gcd(b,a) = \gcd(a,r)$ and that $r$ has at most $\log a + 1$ bits.
    Thus, even if one of the numbers is very large, we can reduce it to size at most $a$ in a constant number of steps.

    Next, we show how to ``shave off'' at least one bit in each iteration.
    First, check whether both numbers are even.
    This can happen only at the beginning, so we divide both numbers by two until at least one of them is odd.
    We keep track of the number of removed factors of two and multiply this power of two back into the final result.

    Second, check whether exactly one of the two numbers is even.
    If so, divide the even number by two and continue.
    In this step, we clearly shave off one bit.
    Otherwise, both numbers are odd.
    Replace $b$ by $b' = b-a$.
    Then $\gcd(a,b) = \gcd(a,b')$.
    Since both $a$ and $b$ are odd, $b'$ is even, so in the next iteration we can divide by two and again shave off at least one bit.

    We terminate when one of the two numbers is zero, one, or equal to the other.
    At that point, no further reduction is needed.

    This finishes the proof of \Cref{lem:Euclidalgorithm}.
    \end{proof}
    
We finish the proof of \Cref{thm:FastFactoring} with a runtime analysis of aboves procedure.
The binary search needs at most $O(\log n)$ steps to find the appropriate $k$.
Furthermore, in each step we compute $k!$ in at most $O(\log^2 n)$ steps and the greatest common divisor in $O(\log n)$ steps as well.
This leads to a total running time of $O(\log^3 n)$, which is polynomial in the length of the binary representation of~$n$.
\end{proof}

It is not known whether the \textsc{Factoring Problem} can be solved in polynomial time,
and it is also not known to be \NP-hard.
The best known upper bound shows that it can be computed in polynomial time on a quantum computer~\cite{Shor1997}.
There are also other ways in which we can ``cheat'' with the \realRAM,
but all such examples involve operations that, in some way, access the bit representation of real numbers.

The example above and similar examples are the reason why we strictly prohibit any operation that provides access to the binary representation of real numbers.
On the positive side, there are no ``cheating'' examples known once we restrict ourselves in this way.

\paragraph{No Division Needed.}
It might be surprising, but we can simulate division (denoted by $\div$) using only multiplication.
We will later use the fact that our \realRAM has only the three operations $+$, $-$, and $\times$; therefore this is a good time to explain how to simulate division.
\begin{lemma}
  Let $A$ be a program on a \realRAM that uses the operations $+$, $-$, $\times$, and $\div$ and decides some decision problem~$P$.
  Then there is another program~$B$ that decides $P$ in the same asymptotic running time and uses only $+$, $-$, and $\times$.
\end{lemma}
The idea is to represent each rational number $q=a/b$ by storing the pair $(a,b)$.
Initially, an input number $q$ is represented as $(q,1)$.
Later, if we have two numbers $q_1=a_1/b_1$ and $q_2=a_2/b_2$, then
\[
\frac{q_1}{q_2} = \frac{a_1/b_1}{a_2/b_2} = \frac{a_1b_2}{a_2b_1}.
\]
Similarly, we can simulate subtraction via
\[
q_1-q_2 = \frac{a_1}{b_1}-\frac{a_2}{b_2} = \frac{a_1b_2-a_2b_1}{b_1b_2}.
\]
Finally, to test whether $q>0$ for $q=a/b$, it suffices to check whether $a\cdot b>0$.

\subsection{Relation between Real and Discrete Computations.}
One might wonder how the \realRAM and the \wordRAM relate to one another. 
Surprisingly, we can give a very precise answer to this question as far as polynomial-time computations of decision problems are concerned.
To describe this relation, we have to introduce two concepts: computational models with access to an oracle, and straight-line programs.

\paragraph{Oracle access.}
A standard way to increase the computational power of a model is to equip it with access to an oracle. Informally, an oracle can be viewed as a black-box function that can be queried in a single computational step. 
For any input, it returns either yes or no.

\paragraph{Straight-line programs.}
When we write a program, we typically use loops and conditional statements.
If we want to compute a number without loops or conditionals, we refer to this as a straight-line program.
More precisely, a \emph{straight-line program} is a finite sequence of assignments
\[
a_0 = 1 \quad \text{and} \quad 
a_k = a_{i} \,\circ \, a_{j},
\quad \text{with } i,j < k \text{ and } \circ \in \{+,-,\times\}.
\]
Clearly, we can think of each $a_i$ both as a program describing how to compute the value and as the value itself.
The problem \textsc{Positive Straight-Line Program} (\PosSLP) asks if a given straight-line program computes a positive number or not.
At first glance this problem seems deceptively easy.
We give the following example that shows that we cannot easily compute the straight-line program on a \wordRAM.

Let 
\[ a_0 = 1, \quad a_1 = 1+1 = 2, \quad \text{and}\]
\[ a_{k+1} = a_k \times a_k,\quad  k = 1,\ldots, n.\]
An easy computation shows that \[a_k = 2^{2^k}.\]
Now it is clear that $a_n$ is positive, but it is also clear that the number has $2^n$ many bits in its binary representation and thus cannot be even stored in polynomial space on a \wordRAM.
Now things become complicated if our straight-line program computes $2^{2^{n_1}} - 3^{2^{n_2}} - 5^{2^{n_3}}$ for some $n_1,n_2,n_3$. 
Then it is not so easy to see if this is a positive number or not.
And we can imagine many even more complicated situations easily.

As a matter of fact, we do not know much of the complexity of \PosSLP and it seems very hard to understand the precise nature of this algorithmic problem.
It is known that \PosSLP is in \PSPACE and thus decidable, and there are some conditional NP-hardness proofs, but those conditions are maybe not true~\cite{burgisser2023hardness, perrucci2007realroots}.

Yet it plays a crucial role in understanding the interplay between real and discrete computations.

\paragraph{Simulating \RealRAM Computations.}
We have now the tools to formulate and prove the following theorem~\cite{allender2009complexity}.
We denote by $\text{P}^\PosSLP$ the set of languages computable on a \wordRAM in polynomial time with oracle access to $\PosSLP$
and with $\PR$ the set of Boolean languages computable in polynomial time on a \realRAM.

\begin{theorem}
    $\text{P}^\PosSLP = \PR$.
\end{theorem}

\begin{proof}
    It is easy to see that $\text{P}^\PosSLP \subseteq \PR$, as the \realRAM can compute \PosSLP easily.

    The reverse direction is a bit more difficult. 
    We assume that the \realRAM has only the basic operations $+$, $-$, and $\times$ to operate on the real numbers.
    Let $A$ be an algorithm on the \realRAM that we want to simulate in polynomial time with some algorithm $B$ on the \wordRAM.
    Now, our algorithm $B$ on the \wordRAM stores not the result of each computation, but the straight-line program of each real number computation.
    (Note that although the numbers are real, they are computed from Boolean input.)
    Otherwise, $A$ and $B$ are doing exactly the same steps.
    Now, it could be that $A$ has a comparison operation 
    \[ \text{If } x> 0 \text{ then } \ldots \]
    If we want $B$ to do exactly the same as $A$, it needs to be able to evaluate this comparison.
    Therefore, we replace this part of the program by taking the straight-line program that computes $x$ and then calling our \PosSLP oracle to resolve the comparison.
    In summary, $B$ does exactly the same as $A$ and also branches in the same way as $A$ using \PosSLP. 
    Therefore, $A$ and $B$ will return the same yes/no answer for the given decision problem.

    It remains to argue that $B$ runs in polynomial time.
    As $A$ runs in polynomial time, the length of the straight-line program is also of polynomial length.
    Thus, each step takes at most a polynomial number of additional steps in $B$ compared to $A$.
    As the product of two polynomials is a polynomial, the runtime of $B$ is still polynomial.
\end{proof}

To summarize, the problem \PosSLP encapsulates precisely the difference of the real and discrete computations as far as polynomial time computation of decision problems is concerned.

\subsection{Exercises}
These exercises are meant to review and deepen the content of this section.

    \begin{enumerate}[noitemsep]
        \item Which operations would you allow on a word RAM? 
        Make a complete list and weigh the pros and cons of each operation.
        \item It was mentioned that registers are allowed to store words of size $\Theta(\log n)$, where $n$ is the input length.
        Why do we need a word length of at least $\log n$?
        What would go wrong if we did not?
        What would happen if we allowed words of size $n$? 
        \item Give a simple \wordRAM program that can find an element $x$ in a sorted list of integers.
        Use only primitive operations. 
        \item Try to give a formal definition of a Turing machine as a mathematical object. 
        The more general you are, the more difficult it might become. Try to find a simple, workable definition. 
        \item Show that if we are only interested in decision problems then we can remove the subtraction operation from the real RAM.
        However, we need to ensure that we can compare two numbers $a<b$, rather than only performing comparisons of the form $a > 0$.
        \item Let $n$ be an integer that needs $b$ bits to be represented in binary. Show that there is a straight-line program of size $O(b)$ that computes $n$. 
        \item It is possible to allow different operations in straight-line programs. 
        For example, we can only allow the operations $O_1 = \{ + , -, \fbox{\rule{0pt}{0.1cm}\rule{0.1cm}{0pt}}^{2} \}$, 
        where $\fbox{\rule{0pt}{0.1cm}\rule{0.1cm}{0pt}}^2 : x \mapsto x^2/2$.
        Another example are the operations 
        $O_2 = \{ + , -, \inv \}$.
        Here $\inv : x \mapsto 1/x$, if $x \neq 0$.
    \begin{enumerate}
        \item 
        Show that for every straight-line program computing $x$ in $n$ steps, there exists another straight-line program computing $x$ in $O(n)$ steps using the operations $O_1$.\\
        Hint: \textcolor{white}{Rewrite $(a+b)^2 = a^2 + 2ab + b^2$.}
        \item Show that for every straight-line program computing $x$ in $n$ steps, there exists another straight-line program computing $x$ in $O(n)$ steps using the operations $O_2$.
        
        Hint: \textcolor{white}{Rewrite $\frac{1}{x} -\frac{1}{x+1} = \frac{1}{x^2 + x}$}
    \end{enumerate}
    \end{enumerate}

\paragraph{Open Pool Exam Questions.}
These questions capture the learning goals of the section.
\begin{enumerate}[noitemsep]
        \item Explain a program on a \wordRAM. (Not a formal definition.)
        \item Explain a Turing machine. (Not a formal definition.)
        \item Give a comparison of advantages and disadvantages of the two models.
        \item Explain what a program on a \realRAM is and give both a motivation and concerns about this model of computation.
        \item Show how to compute $k!$ in $O(\log^2 k)$ time on a \realRAM that is allowed to use rounding.
        \item Show how to compute the greatest common divisor $\gcd(a,b)$ of two numbers in $O(\min \log a, \log b)$ time on a \realRAM that is allowed to use rounding.
        \item Show that we can find a factor of a number in polynomial time on the \realRAM that is allowed to use rounding. You are allowed to use the the answer of the previous two questions.
        \item Define a straight line program and \PosSLP.
        \item Show that $\text{P}^{\PosSLP} = \text{P}\R$.
    \end{enumerate}

\newpage
\section{Existential Theory of the Reals}

\subsection{Motivation}
The \emph{existential theory of the reals} gives rise to a complexity class abbreviated as \ER (pronounced ``ER'' or ``exists R'').
It can be viewed as a ``real-number analogue'' of \NP: instead of binary witnesses, we allow witnesses over $\R$, and the verifier is allowed to perform exact real arithmetic.

This viewpoint is useful because many natural decision problems in geometry are \NP-hard, yet they do not come with an obvious \NP certificate.
Often, a solution exists but cannot be described with a small number of bits (for example, it may require irrational coordinates).
The class \ER captures precisely this phenomenon: problems in \ER have polynomial-time verifiers with real-valued witnesses, and \ER-complete problems form a robust hardness notion for such geometric questions.

Historically, a growing body of geometric and topological problems was shown to be complete for this class, and Marcus Schaefer unified many of these results under the name ``$\exists\mathbb{R}$'' in 2010~\cite{S10}.
First, we give two definitions of the class \ER.
Then, we place \ER in relation to other known complexity classes.


\subsection{Definitions}

\paragraph{Via Logic.}
The algorithmic problem \etr stands for \emph{existential theory of the reals}. It is a decision problem about systems of polynomial constraints over real numbers.
An instance consists of a finite set of real variables
$x_1, x_2, \dots, x_n$
and a finite collection of polynomial constraints of the form
\[
p(x_1,\dots,x_n) \ R \ 0,
\]
where each \(p\) is a polynomial with integer coefficients and $R \in \{<,\leq,=,\geq , >\}$.
We assume here that the polynomial is given in standard form.
The constraints can be combined using logical operations such as “and”, “or”, and “not”.
(Equivalently, one can assume the formula is in a normal form using only conjunctions and disjunctions of such atomic constraints.)
A simple example is 
\[
\exists x,y :
x^2 + y^2 \le 1
\ \land \
x \ge 1/2.
\]
It asks for real numbers $x,y$ such that they are in the unit disk and $x$ being at least a half.
\begin{center}
    \includegraphics{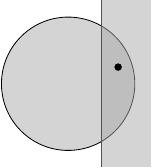}
\captionof{figure}{Here, you can see the region in the plane that satisfies both of those constraints.}
\end{center}

Now, the complexity class \ER is the set of problems that reduce in polynomial time to \etr.
Here reductions are ordinary reductions on a standard discrete machine model (a Turing machine or a \wordRAM), so the input is still a finite binary string.
The real numbers enter only through the quantified variables of the \etr instance.

\paragraph{Via Machine Models}
When we think of the complexity class \NP, we have two definitions:
one that says every problem that reduces to \satisfiability,
and one that says every decision problem $P$ is in \NP if the following holds.
There exists a polynomial-time algorithm $A$ running on a \wordRAM such that 
\[ I \in P \ \Leftrightarrow  \ \exists w\in \{0,1\}^* : A(I,w) = 1.\]
Here, by $I \in P$ we mean that $I$ is a yes-instance of $P$, and $w$ is typically referred to as a witness.
The set $\{0,1\}^* = \bigcup_{i\in \N} \{0,1\}^i$ is the set of all Boolean strings. 
Similarly, the set $\R^* = \bigcup_{i\in \N} \R^i$ is the set of all finite real-valued vectors of any dimension.

In the same way we can define \ER. We formulate the following theorem.
\begin{theorem}
    \label{thm:Machine-Model}
    A problem $P$ is in \ER if and only if the following holds.
There exists a polynomial-time algorithm $A$ running on a \realRAM such that, for every binary-encoded instance $I$,
\[ I \in P \ \Leftrightarrow  \ \exists w\in \R^* : A(I,w) = 1.\]
\end{theorem}
Note that compared to \NP, the only difference is that the verifier runs on a \realRAM and that the witness $w$ is allowed to be real-valued instead of binary.
In our model of the \realRAM, the instruction set is fixed and the only available real constants are $0$ and $1$.

It is difficult to say when exactly this theorem was first established.
The proof ideas go back to the original proofs of Cook~\cite{Cook1971} and Levin~\cite{Levin1973}.
For real models of computation this was first adapted by Blum, Shub and Smale~\cite{BlumShubSmale1989}.
However, their real model of computation does not exactly match ours.
Another proof can be found by Erickson, van der Hoog, and Miltzow~\cite{EricksonVanDerHoogMiltzow2024, EvdHM20}, who work explicitly with a \realRAM model.
This theorem is clearly easier to prove for the real Turing machine than for the \realRAM, because real Turing machines allow fewer operations.
Conceptually, both proofs are identical though.

\begin{proof}[Proof Sketch.]
    There are two directions.
    First, we show how an \etr-formula fits into the machine-based definition.
    Second, we show how to express a polynomial-time \realRAM computation as an \etr-formula.

    For the first direction, we take the existentially quantified real variables of the \etr-formula as the witness $w$.
    Evaluating the formula is a standard task for the \realRAM: it can compute the involved polynomials using $+$, $-$, and $\times$, and it can combine the resulting comparisons using the Boolean structure of the formula.

    For the second direction, the idea is to use \etr as a very low-level programming language.
    Fix an input instance $I$ of size $n$ and a running-time bound $T(n)$ for the \realRAM algorithm $A$.
    We introduce variables that represent the contents of each register at each time step $t\in\{0,1,\dots,T(n)\}$, together with variables encoding the current instruction.
    See the following illustration.
    \begin{center}
        \includegraphics{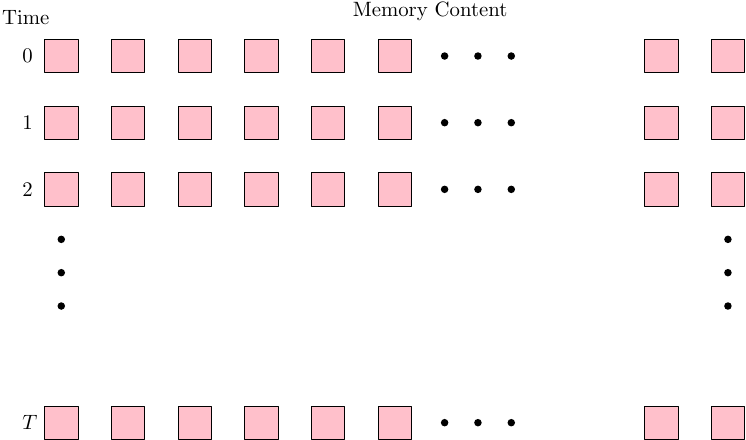}
        \captionof{figure}{The Cook--Levin tableau idea: we encode the contents of memory cells (and the current instruction) at each time step and constrain consecutive time steps to follow the program.}\label{fig:cookLevinTableau}
    \end{center}
    Then we write constraints that enforce:
    (i) \texttt{Initiate}: the time-$0$ registers contain the input encoding,
    (ii) \texttt{Update}: each step follows one legal transition of the \realRAM instruction set, and
    (iii) \texttt{Accept}: at some time step the computation is in an accepting state.
    The resulting \etr-formula is satisfiable if and only if there exists a witness $w\in\R^*$ such that $A(I,w)=1$.
\end{proof}

\paragraph{Why a Machine Model?}
The above theorem is interesting for several reasons.
First, it is a useful tool when we want to show \ER-membership for a concrete problem.
Researchers who are very skilled at manipulating logical formulas might find it less helpful: they already have a good intuition for what can and cannot be expressed by \etr-formulas.
For readers less familiar with this logic-based viewpoint, however, reasoning in terms of a \realRAM verifier can be much more accessible.

Second, it explains why the \realRAM is a natural extension of the \wordRAM when moving from \NP to \ER.
In this sense, the \realRAM is not an arbitrary choice of model, but an answer to the question: what is the ``right'' machine model whose witnesses are real-valued?

Finally, when studying the power of \ER it is often convenient (and sometimes necessary) to formulate conjectures and theorems in terms of a machine model.
For example, this lets us talk about oracle access to \ER and compare it with oracle access to \NP.

\subsection{Comparison to known complexity classes.}
It is known that $\NP \subseteq \ER$ and that $\ER \subseteq \PSPACE$.
The first inclusion is easy, and we give one possible short proof here.
The second inclusion is highly non-trivial and goes back to Canny~\cite{Canny1988RobotMotionPlanning, Canny1988PSPACEAlgebraicGeometric}.
We show here the first inclusion.

\begin{lemma}
  $\NP \subseteq \ER$.
\end{lemma}
\begin{proof}
  We reduce \textsc{3-Coloring} to \etr.
  Let $G=(V,E)$ be a graph with $V=\{1,\ldots,n\}$.
  We construct an existential formula over the reals that is satisfiable if and only if $G$ is $3$-colorable.

  Introduce one real variable $x_v$ for every vertex $v\in V$.
  We intend that $x_v$ encodes the color of $v$.
  We encode the colors using the three numbers $0$, $1$, and $2$.
  Concretely, we enforce
  \[
    p(x_v) := (x_v-0)(x_v-1)(x_v-2) = 0.
  \]
  Now we need to enforce that adjacent vertices receive different roots.
  For $x,y$ that represent adjacent vertices, we add the constraint $x \neq y$.
  This completes the reduction from \textsc{3-Coloring} to \etr.
  Since every problem in \NP reduces to \textsc{3-Coloring}, every problem in \NP also reduces to \etr and is therefore contained in \ER.
\end{proof}

The proof we just gave is also interesting, as it shows how \etr-formulas can encode different problems.
We can give a second proof by noting that $\{0,1\}^* \subseteq \R^*$ and that a \realRAM algorithm can simulate every \wordRAM algorithm.
Thus, if an algorithm $A$ shows that some problem $P$ is in \NP, we can use the same algorithm to show that $P$ is in \ER, with the extra check that the witness $w$ is binary.

\subsection{Exercises}
These exercises are meant to review and deepen the content of this section.
 \begin{enumerate}[noitemsep]
        \item 
        Write a logical formula $\varphi$ such that the solution space 
        $S = \{x\in\R^2 : \varphi(x)\}$ forms this geometric shape.
        
        \begin{center}
            \includegraphics{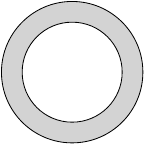}    
        \end{center}

        \item Write a formula $\varphi$ such that $\varphi(x)$ is true if and only if $x$ is an integer between $1$ and $n$.
        (Can you write $\varphi$ with $O(\log(n))$ symbols?)

        \item  Let $U_1,\ldots,U_n$ be the variables and $x$ be an integer between $1$ and $n$.
        Write a formula $\varphi$ such that $\varphi(x,U_1,\ldots,U_n,z)$ is true if and only if $z = U_x$. 
        Note that it is forbidden to write $U_x$ in a formula. This exercise captures the main idea for simulating indirect addressing on the real RAM using logical formulas, which is the most difficult part of \Cref{thm:Machine-Model}.

        \item Show that geometric packing is in \ER. 
        In geometric packing, we are given $n$ polygonal pieces and a polygonal container in the plane. 
        We are allowed to translate and rotate the pieces such that no two pieces overlap and all of the pieces are contained in the container.
        Give one proof using the logic-based definition of \ER and one proof using the machine-model definition of \ER.

        \item Show that the art gallery problem is in \ER. 
        In the art gallery problem, we are given a polygon and a number $k$ and we are asked to find $k$ guards (points) inside the polygon such that every point inside the polygon is seen by at least one of the guards.
        (A guard $g$ sees a point $p$ if and only if the line segment $gp$ is contained in the polygon.)
        Give one proof using the machine-model definition of \ER.
        If you are allowed to use universal (``for all'' $\forall$) quantifiers, how could you encode solvability of the art gallery problem?
        (Warning: It is possible to show that the art gallery problem is in \ER using an \etr formula. It took me several days to find one and the proof spans 5 pages. This was part of the original motivation for the machine model definition of \ER.)
        \item Show (using the logical definition of \ER) that $\textsc{PosSLP} \in \ER$.
        \item In $\textsc{Euclidean Travelling Salesperson Problem}$, we are given $n$ points in the plane and a number $t\in \mathbb{Q}$.
        We are asked to find a route that visits all $n$ points and is shorter than $t$.
        \begin{enumerate}[noitemsep]
          \item Show that the problem lies in \ER.
          \item It is unknown whether the $\textsc{Euclidean Travelling Salesperson Problem}$ is in \NP. What could be the reason?
          Feel free to speculate. A one-line answer suffices.
          \item It is unknown whether the $\textsc{Euclidean Travelling Salesperson Problem}$ is \ER-hard. 
          What might be the reasons? Feel free to speculate. A One-line answer suffices.
        \end{enumerate}
    \end{enumerate}

\paragraph{Open Pool Exam Questions.}
These questions capture the learning goals of the section.
 \begin{enumerate}[noitemsep]
        \item Give a definition of the algorithmic problem \etr and the complexity class \ER based on this definition.
        \item Give the machine model description of \ER and a sketch of a proof that it is equivalent with the logic based definition.
        \item Show that $\NP \subseteq \ER$.
    \end{enumerate}

\newpage
\section{\ordertype and \stretchability}
Numerical instability was an issue that computational geometry researchers noticed early on. 
There was a desire to find algorithms that can solve geometric problems but are in nature combinatorial.
Interestingly many geometric problems can be solved efficiently solely by knowing the order type of a point set.
So what is the order type of a point set.

\begin{center}
    \includegraphics[]{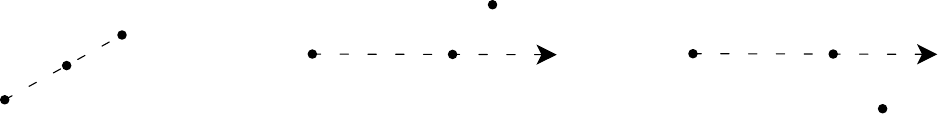}
\end{center}
It specifies, for each ordered triple of points, whether the third point lies on, to the left of, or to the right of the line through the first two points.
Equivalently, for each triple of points it specifies whether the points are collinear or oriented clockwise or counterclockwise.
For example, it is easy to see that the points on the convex hull of a point set are completely determined by the order type.
Thus, once we can extract the order type of a point set, we can use it to compute the convex hull without ever looking at the coordinates again.
Another motivation to study order types is the desire to prove more structural properties of point sets, analogous to how we know many structural properties of graphs.
For example, if we accept that order types contain all essential information about a point set, then for each $n$ there exist only finitely many order types (and hence only finitely many point configurations up to order type).

Out of mathematical curiosity Ringel asked the following question~\cite{Ringel1956} in 1956: ``Are there two different point sets $P,Q$ in the plane such that we cannot continuously move one to the other without changing the order type?''
The surprising answer is that there are indeed such point sets, but it is difficult to give a concrete example or a simple direct proof.
The currently smallest known such point set has 13 points and was found in 2012~\cite{Tsukamoto2012Disconnected}.

To answer Ringel's question Mnëv~\cite{M85} and Shor~\cite{S91} studied the realization space of \ordertype{}s in full generality, see also~\cite{RichterGebert1995}.
Specifically, they formulated it in the equivalent language of \stretchability{}.

\begin{center}
    ``Mnëv proved an even stronger theorem. When translated into complexity theory terms, his theorem implies that determining the stretchability of pseudoline arrangements is equivalent to the existential theory of the reals.'' Shor 1991~\cite{S91}.
\end{center}
Interestingly this statement was done in 1991, around 20 years before the complexity class \ER was named.
The reduction is generally attributed to Mnëv, but it was Shor who made the proof more accessible to a general audience and who also saw the computer science perspective of Mnëvs work.
For an even more streamlined and very readable modern exposition, see Matou{\v{s}}ek’s lecture notes~\cite{Matousek2014ETR}.

The \textsc{Order Type Problem} asks whether, for a given order type, there exists a point set with that order type.
Although this problem might look a bit strange at first, it is closely intertwined with computational geometry and the existential theory of the reals.
The \textsc{Order Type Problem} was not only the first problem shown to be \ER-complete; it also serves as a stepping stone for many other \ER-hardness proofs.
Furthermore, point sets in the plane are among the most fundamental objects in computational geometry and have had a major influence on the field.

\subsection{Partial Order Types}
To give an idea of the flavor of the proof, we focus on partial order types.
A \emph{partial order type} is like an order type, except that we specify the orientation only for some triples.
We define the \partialOrderTypeProblem as follows.
We are given a partial order type, and we ask whether there exists a point set in the plane that agrees with it on the specified triples.

\begin{theorem}
  \label{thm:Order-type-ER}
  The \partialOrderTypeProblem is \ER-complete.
\end{theorem}

The proof goes in two steps.
We first reduce \etr to the more convenient problem \textsc{ETR-AM}.This part is purely on the symbolic level of manipulating formulas in a smart way.
Thereafter, we reduce to \partialOrderTypeProblem.
The second part of the reduction is geometric in nature.

\paragraph{ETR-AM.}
\textsc{ETR-AM} is simply a conjunction of constraints of the form $x=1$, $x+y=z$ or $x\cdot y=z$.
The letters A and M stand for addition and multiplication.
Note that if we did not have a constraint of the form $x=1$, we could set all variables to zero and have a valid solution.

\begin{lemma}
\textsc{ETR-AM} is \ER-complete.
\end{lemma}
\begin{proof}
We first note membership and then prove hardness.
For membership, note that \textsc{ETR-AM} is in \ER because it is a syntactic restriction of \etr.

To show hardness, let $\Phi$ be an arbitrary \etr-instance.
We construct, in polynomial time, an equivalent instance of \textsc{ETR-AM} by a sequence of standard normal-form transformations.

Push negations down to atomic constraints (e.g., $\lnot(p>0)$ becomes $p\le 0$), and distribute disjunctions until they occur only at the level of polynomial equalities/inequalities.
Next, remove (strict) inequalities using a fresh variable; for example,
\[
 p>0 \quad\Longleftrightarrow\quad \exists x\colon p\cdot x^2 = 1.
\]
and similarly 
\[p\geq 0  \quad\Longleftrightarrow\quad \exists x \colon p =x^2.\]
Disjunctions of equalities can then be removed by multiplication; for instance,
\[
(p=0)\ \lor\ (q=0) \quad\Longleftrightarrow\quad p\cdot q = 0.
\]
After this step, the formula is purely existential and consists of a conjunction of polynomial equations.

Next we remove coefficients.
We introduce auxiliary variables that build the required constants from $1$ using additions and multiplications.
For example, writing $2=b$ via $b=1+1$ and $10=e$ via repeated doubling/addition.
For example $4 = c = b+ b$, $8 = d = c + c $, and $10 = e = d + b$.
This allows us to replace an equation such as $10x^2=0$ by an equivalent system that uses only constants $1$ and addition.

 At last, we arithmetize each polynomial into $\{x+y=z,\; x\cdot y=z\}$ constraints.
For each product or sum appearing in a polynomial, introduce a fresh variable for its value and add the corresponding constraint.
To express a final constraint of the form $c=0$ using only $x=1$ and additions, replace it by $c+1=1$.
For instance, $x^2y+z=0$ can be encoded as
\[
\exists a,b,c\colon\ a=x\cdot x\ \land\ b=a\cdot y\ \land\ c=b+z\ \land\ c+1=1.
\]

The resulting formula is a conjunction of constraints of the form $x=1$, $x+y=z$, and $x\cdot y=z$, i.e., an instance of \textsc{ETR-AM}.
This reduction is polynomial-time and preserves satisfiability, proving \ER-hardness.
\end{proof}

\paragraph{Geometric Reduction.}
We are now ready to describe the first part of the geometric reduction.
We start with two points, denoted ${\bf 0}$ and ${\bf 1}$, and think of them as representing the values zero and one. 
Then, for each variable of the \textsc{ETR-AM} instance, we place a point on the line defined by ${\bf 0}$ and ${\bf 1}$.
We can easily enforce this using the partial order type information that we are allowed to specify.
The ratio $a= \frac{|{\bf x}-{\bf 0}|}{|{\bf 1}-{\bf 0}|}$ gives the value of the variable corresponding to the point ${\bf x}$.
Equivalently, there is a linear transformation that maps ${\bf 0}$ to the origin $(0,0)$ and ${\bf 1}$ to the point $(1,0)$.
This maps ${\bf x}$ to the point $(a,0)$, where $a$ is the value that ${\bf x}$ represents.

Then, for each addition or multiplication constraint, we build the corresponding gadget.
We replace every point corresponding to a variable set to the value $1$ by the point ${\bf 1}$.
Note that we do not need to worry about how two distinct gadgets interact, since we do not need to specify a full order type, but only a partial order type.
To describe those addition and multiplication gadgets, it is useful to first assume that we are also allowed to enforce that lines are parallel.
The constructions are described in the figure below.
Note that all lines in the figure that are parallel are actually enforced to be parallel.
\begin{center}
    \begin{minipage}{\linewidth}
      \centering
      \includegraphics{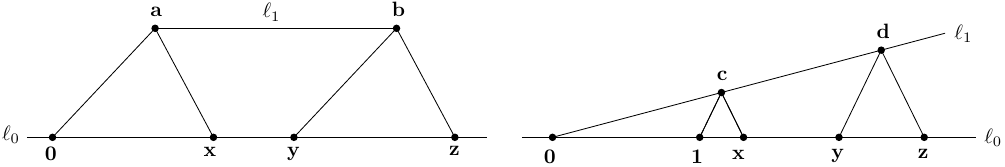}
      \captionof{figure}{Addition and multiplication gadgets with parallel lines.}\label{fig:Staudt-Construction}
    \end{minipage}
\end{center}
We first describe the addition constraint on the left.
We will use the convention that a variable $x$ is represented by the point $\bf x$, simply denoted by the same bold letter.
The lines ${\bf 0} {\bf a}$ and ${\bf y}{\bf b}$ are parallel; similarly, the lines ${\bf a}{\bf x}$ and ${\bf b} {\bf z}$ are parallel.
Finally, the lines ${\bf a}{\bf b}$ and ${\bf 0}{\bf 1}$ are parallel.
It follows that the triangles formed by ${\bf 0},{\bf x},{\bf a}$ and ${\bf y},{\bf z},{\bf b}$ are actually identical.
Thus, $\|{\bf 0} - {\bf x}\| = \|{\bf y} - {\bf z}\|$, which implies that $x + y = z$ under this interpretation.

Now we explain multiplication.
The gadget is shown on the right of \Cref{fig:Staudt-Construction}.
We identify several pairs of similar triangles.
Recall that similar triangles are identical up to scaling.
Specifically, we use that two triangles are similar if their corresponding sides are parallel.

\noindent
The triangles ${\bf 0}{\bf 1}{\bf c}$ and ${\bf 0}{\bf y}{\bf d}$ are similar. So,
\begin{equation}\label{eqn:1}
\frac{|{\bf c} - {\bf 1}|}{|{\bf 1} - {\bf 0}|}
=
\frac{|{\bf d} - {\bf y}|}{|{\bf y} - {\bf 0}|}.
\end{equation}

\noindent
The triangles ${\bf 0}{\bf x}{\bf c}$ and ${\bf 0}{\bf z}{\bf d}$ are similar. So,
\[
\frac{|{\bf c} - {\bf x}|}{|{\bf x} - {\bf 0}|}
=
\frac{|{\bf d} - {\bf z}|}{|{\bf z} - {\bf 0}|}.
\]
Triangles ${\bf 1}{\bf x}{\bf c}$ and ${\bf y}{\bf z}{\bf d}$ are also similar, so
\[
\frac{|{\bf c} - {\bf x}|}{|{\bf c} - {\bf 1}|}
=
\frac{|{\bf d} - {\bf z}|}{|{\bf d} - {\bf y}|}.
\]
This implies that
\begin{equation}\label{eqn:2}
\frac{|{\bf c} - {\bf 1}|}{|{\bf x} - {\bf 0}|}
=
\frac{|{\bf d} - {\bf y}|}{|{\bf z} - {\bf 0}|}.
\end{equation}
Combining Equation~\ref{eqn:1} and \ref{eqn:2} gives
\[\frac{|{\bf x} - {\bf 0}|}{|{\bf 1} - {\bf 0}|} 
= \frac{|{\bf z} - {\bf 0}|}{|{\bf y} - {\bf 0}|}.
\]
This is equivalent to $x/1 = z/y$, and thus $x y = z$, as claimed.

\smallskip

The next step is to modify this reduction so that we no longer require parallel edges. This can be done using the projective plane.

\paragraph{Projective Plane.}
One intuitive way to explain projective transformations is by looking at the visual effect they have on drawings.
\begin{center}
    \includegraphics[page = 1]{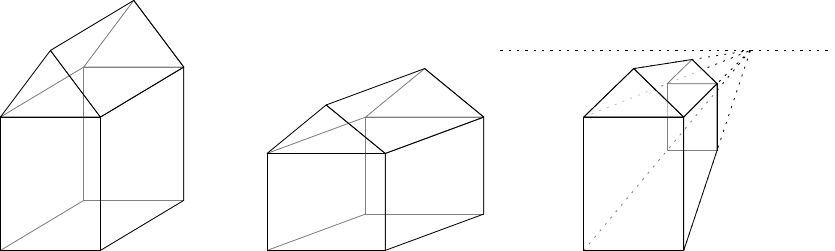}
    \captionof{figure}{On the left, we have a stylized image of a simple three-dimensional house.
    In the middle, we have rescaled the image, and on the right, we applied a projective transformation.
    Note that all edges that are parallel in reality are also drawn parallel.
    Now, if we were to look at the house with our own eyes (and, somewhat strangely, even see through walls), we would see parallel lines meet at the horizon.}
\end{center}
One might wonder how we can model those transformations mathematically; the answer is the projective plane.
The projective plane is like the Euclidean plane, but with a line at infinity added, as if it were any other line.
Now, linear transformations in this projective plane are exactly what we think of as projective transformations.
Crucially, they map lines to lines and preserve incidence and collinearity: if a point lies on a line before the transformation, then its image lies on the image of the line, and if three points are collinear before, then their images are collinear after.
And if those transformations keep the line at infinity at infinity, then they are just ordinary linear transformations.
Specifically, we can transform any line $\ell$ in the plane to the line at infinity.
This has the effect that if two lines $i,j$ meet at a point on $\ell$, then after the transformation $i$ and $j$ will be parallel.
For the following, this is all we need to know about projective lines.
We summarize this in the following lemma.

\begin{lemma}[Projective Transformations]
  A projective transformation of the plane maps lines to lines and preserves incidence and collinearity.
  In particular, if two lines $i$ and $j$ intersect at a point $p$ and $p$ is mapped to a point at infinity, then the images of $i$ and $j$ are parallel.
\end{lemma}

\paragraph{Geometric Reduction Continues.}
We are now ready to descirbe the geometric reduction without the assumption that we can construct lines to be parallel. The idea is tha we create a line $\ell_\infty$ into the Euclidean Plane and treat it as if it were at infinity.
When we apply later an projective transformation to move $\ell_\infty$ to infinity, we end up with the construction as described above.

The following image shows how the addition gadget is constructed.
Note that this configuration can be easily enforced using order type information alone.

\begin{center}
    \includegraphics[page = 2]{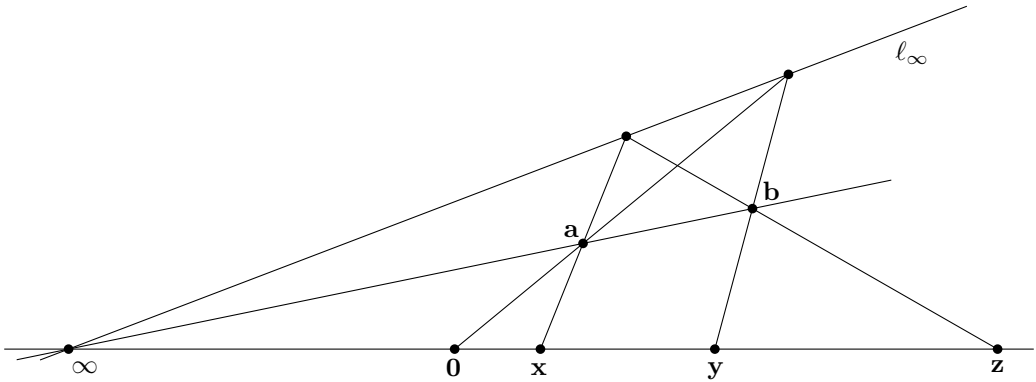}
    \captionof{figure}{Addition gadget without parallel lines.}\label{fig:Non-Parallel-Addition}
\end{center}

The next image shows how the multiplication gadget is constructed.

\begin{center}
    \includegraphics[page = 3]{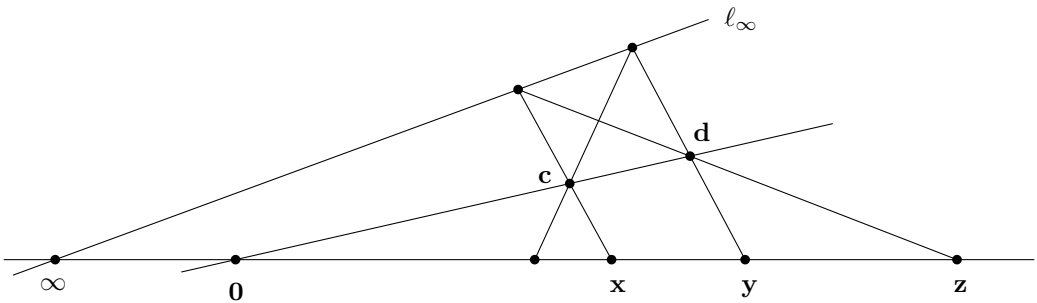}
    \captionof{figure}{Multiplication gadget without parallel lines.}\label{fig:non-parallel-multiplication}
\end{center}

It is now easy to check that if we move the line $\ell_\infty$ to infinity using a projective transformation that the lines that are meant to be parallel in the original addition and multiplication gadgets will be parallel indeed.

\paragraph{Summary.}
Let’s recall the main steps of the proof of \Cref{thm:Order-type-ER}.
First, we reduce from \etr to \textsc{ETR-AM}.
Then we encode each variable as a point on a line, which we can enforce using order types.
Using the two special points $\mathbf{0}$ and $\mathbf{1}$, we can interpret the placement of each point as the value of the corresponding variable.
Assuming that we can enforce lines to be parallel, we then construct gadgets that enforce addition and multiplication.
Finally, we give some basic information on projective transformations and explain how to avoid parallel lines.
We do so by introducing a special line $\ell_\infty$; whenever we want two lines to be parallel, this is equivalent to requiring that they meet at a point on $\ell_\infty$.
If we project $\ell_\infty$ to the line at infinity in the projective plane, the variable values then take the intended values.

\paragraph{Full Simple Order Type.}
The proof we presented here is particularly simple, since we do not have to deal with two obstacles.
First, for most triples of points we do not know their order type, and we also do not have to specify it; gadgets can intersect one another in unforeseen ways.
If we want to show \ER-hardness of full order-type realizability, we need to carefully separate all gadgets and copy and paste variables so that no two gadgets intersect in unintended ways.
The second obstacle is to get rid of collinearities.
Our reduction relies heavily on collinear points.
Interestingly, if all points are constructed in a certain order, it is possible to replace each collinearity by a small constant-size construction that enforces the existence of a collinear point that is not part of the construction.
We skipped these two complications to make the main ideas of the proof more accessible.

\subsection{\stretchability}

In this section, we explain the close connection between \emph{order type realizability} and \emph{stretchability of pseudoline arrangements}.

\paragraph{Pseudoline Arrangements.}
A \emph{pseudoline} is a simple curve in the plane that is unbounded in both directions (topologically, it behaves like a straight line).
A \emph{pseudoline arrangement} is a collection of pseudolines such that any two pseudolines intersect in exactly one point, and they cross there.
Two arrangements are \emph{combinatorially equivalent} if, for every pseudoline, the order in which it meets all other pseudolines is the same in both arrangements.
Algorithmically, we need to agree on an encoding of pseudoline arrangements, which is typically done as follows.
We are given the initial order on the left, from top to bottom, and then a left-to-right specification of the order in which the pseudolines cross.
\stretchability asks whether a given pseudoline arrangement is combinatorially equivalent to an arrangement of straight lines (that is, whether it can be \emph{stretched}).

\begin{center}
    \includegraphics{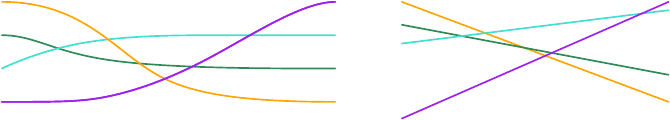}
    \captionof{figure}{A pseudoline arrangement (left) and a combinatorially equivalent straight-line arrangement (right). The \stretchability problem asks whether a given pseudoline arrangement can be stretched to straight lines while preserving the crossing order.}
    \label{fig:stretchability}
\end{center}

We will discuss the following theorem.
\begin{theorem}
  \stretchability is \ER-complete.
\end{theorem}

The core idea of the proof is to reduce from the \textsc{Order Type Problem}.
Bokowski, Mock, and Streinu~\cite{BokowskiMockStreinu2001} provided a coherent proof in the language of oriented matroids.
 
We will not show this reduction here; instead, we explain the close relationship between point configurations in the plane and straight-line arrangements, known as point--line duality.

\paragraph{Point--Line Duality.}
A standard tool to relate point configurations and line arrangements is point--line duality.
One convenient choice is the map $\chi$ from the plane $\R^2$ to the set of all lines in the plane $\mathcal{L}$, defined by
\[
(a,b) \longleftrightarrow \ell_{(a,b)} : y = ax-b,
\qquad
y = mx+c \longleftrightarrow (m,-c).
\]
We often denote the line $\chi(p)$ by $p^*$ and the point $\chi^{-1}(\ell)$ by $\ell^*$.
There are free graphical calculators (Like \href{www.desmos.com}{Desmos}) with which one can dynamically visualize points and their corresponding dual line in the same image.
To illustrate point--line duality, we inlude a small example of three points and one line.

\begin{center}
  \includegraphics[page=1]{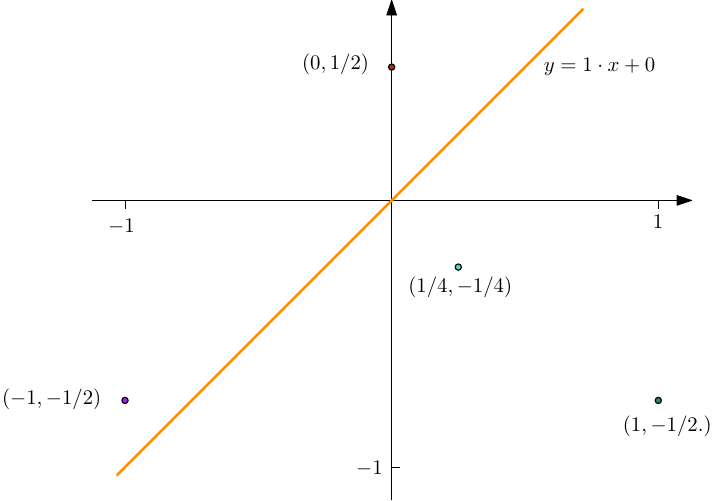}
  \captionof{figure}{Point--line configuration in the primal plane.}
  \label{fig:point-line-duality-primal}
\end{center}
We now apply $\chi$ to translate the geometric situation from points to lines, obtaining the corresponding configuration in the dual plane.
\begin{center}
  \includegraphics[page=2]{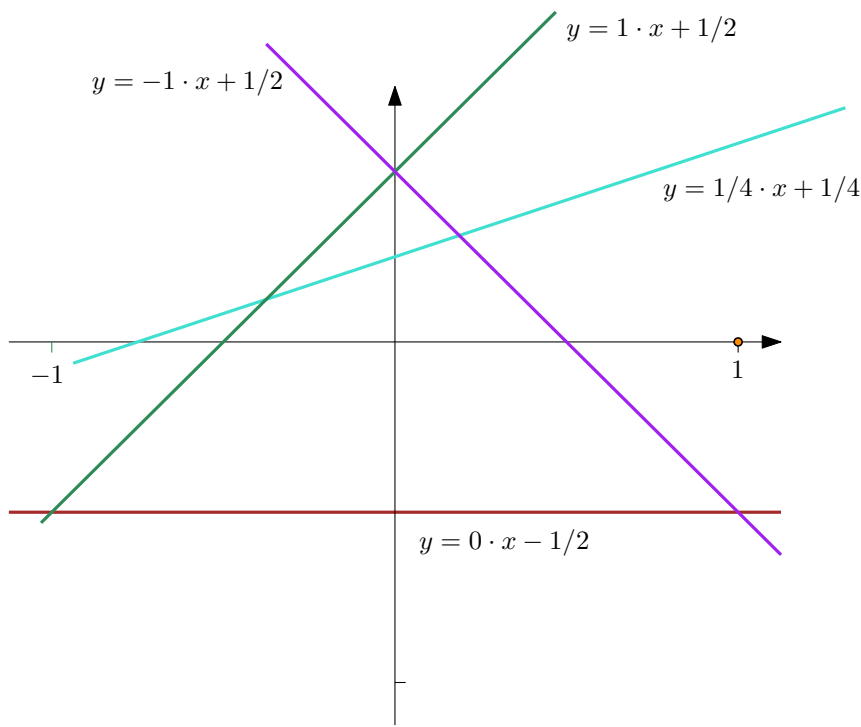}
  \captionof{figure}{The same configuration after applying point--line duality.}
  \label{fig:point-line-duality-dual}
\end{center}
This duality has some key properties, summarized in the following lemma.

\begin{lemma}
  \label{lem:Basic-Duality}
  Let $p,q \in \R^2$ be a point and $\ell \in \L$ be a line in the plane. 
  Then the following holds.
\begin{itemize}
  \item \textbf{verticality:} $p$ and $q$ have the same $x$-coordinate if and only if $p^*$ and $q^*$ are parallel.
  \item \textbf{incidence:} point $p$ lies on line $\ell$ if and only if the dual line $p^*$ passes through the dual point~$\ell^*$.
  \item \textbf{above/below:} $p$ lies above $\ell$ if and only if $\ell^*$ lies above $p^*$.
\end{itemize}
\end{lemma}

\begin{proof}
  The first property follows from the definition of the slope of the dual line.

  Let $p = (a,b)$ and $\ell: y= mx+n$.
  Then $p$ lies on $\ell$ if and only if
  \[b = ma + n\]
  which is equivalent to 
  \[-n = a m -b.\]
  And this is equivalent to $\ell^* = (m,-n)$ lying on the line $p^* : y = ax -b$.
  This proves the incidence statement.

  Similarly, $p$ lies above $\ell$ if and only if
  \[b > ma + n\]
  which is equivalent to 
  \[-n > a m -b.\]
  And this is equivalent to $\ell^* = (m,-n)$ lying above the line $p^* : y = ax -b$.
\end{proof}

As a consequence, combinatorial information about a point set can be translated into combinatorial information about an arrangement of lines, and vice versa.
However, this is not as easy as one might initially hope, because point-line duality preserves above--below relations, whereas order types are about being to the left or to the right of a line.
If we reverse the orientation of a line, above and below stay the same, but a point that was to the left is suddenly to the right.
Also, if we rotate a point set, many above--below relationships change, but the order type is preserved.
Still, many combinatorial properties are preserved under duality.

To make this more concrete, we need the notion of the \textit{cyclic order} of a point set $P$ with respect to a point $p$.
Let $\ell$ be a vertical line that we rotate around $p$ in clockwise direction; the order in which we visit the points in $P$ is the cyclic order of $P$ around $p$.

\begin{lemma}
  \label{lem:cyclic-order}
  Let $p_1,\ldots,p_n \in P$ be the points in cyclic order around $p$.
  The line $p^*$ intersects the lines $p_1^*,\ldots,p_n^*$ from left to right in this order.
  
  Similarly, let $\ell_1,\ldots,\ell_n$ be the lines intersecting a line $\ell$ in this order from left to right.
  Then the dual points $\ell_1^*,\ldots,\ell_n^*$ have the same cyclic order around the dual point $\ell^*$.
\end{lemma}
\begin{proof}
  We only prove the second statement.
  Let $\ell_1,\ldots,\ell_n$ be the lines intersecting the line $\ell$ as in the statement of the lemma.
  We denote by $p_1,\ldots,p_n$ their intersection points.
  See the figure below for an illustration.
  \begin{center}
  \includegraphics[page=1]{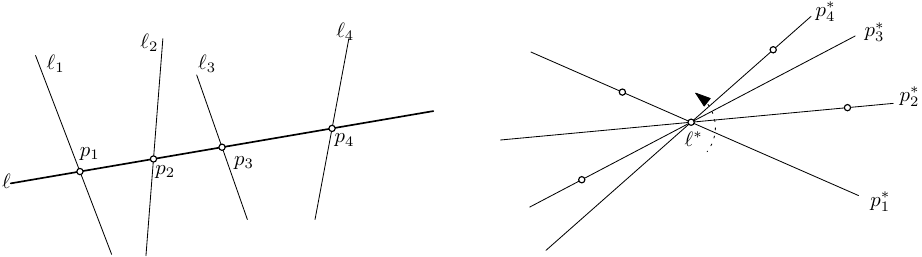}
  \captionof{figure}{Left: The line $\ell$ intersects $\ell_1,\ldots, \ell_n$ from left to right. 
  Right: The cyclic order is the same in the dual.} 
  \label{fig:point-line-duality-dual}
\end{center}
  Clearly, $p_1$ is to the left of $p_2$, and so on, by definition.
  Thus, the line $p_1^*$ has the smallest slope and $p_n^*$ has the largest slope by \Cref{lem:Basic-Duality}.
  Moreover, each dual point $\ell_i^*$ is on the dual line $p_i^*$.
  This shows the lemma.
\end{proof}

\Cref{lem:cyclic-order} gives us a clue on how the reduction from the \OrderTypeProblem to \stretchability. 
As the cyclic sequence determines in which order one pseudoline crosses all other pseudolines. 
Yet it is not enough to construct a complete pseudoline arrangement and then we also need to argue about the equivalence of realizing line arrangement and point sets in the plane.
We omit this here. The central message that we want to bring across is that there is a close connection between points and lines in the plane.
And this connection carries over to the abstract setting although not in a straightforward manner.

\subsection{Geometric Graphs.}
  \stretchability plays a central role in showing that many problems are \ER-complete. Specifically, it has been applied many times to show that recognition of geometric graphs is difficult.
  To understand those techniques, we give two simple examples.
  We first need some definitions.
  Let $W$ be a finite set of objects in the plane.
  We define the intersection graph $G(W)$ of $W$ as follows.
  Every $w\in W$ gives rise to a vertex, and any two vertices are adjacent if the corresponding geometric objects intersect.

\begin{center}
  \includegraphics{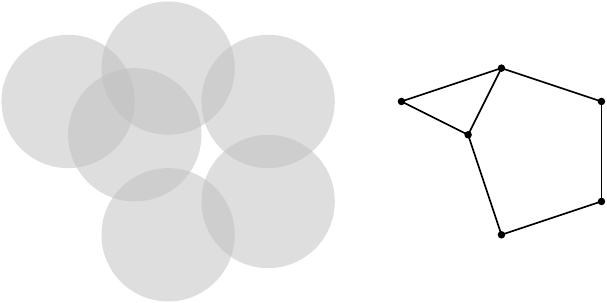}
  \captionof{figure}{Left: a set $W$ of disks; right: the geometric intersection graph $G(W)$.}
\end{center}

\paragraph{Unit Disk Intersection Graphs.}
  We show the following theorem, due to McDiarmid and M\"uller (2013)~\cite{McDiarmidMuller2013IntegerRealizations}.
  They were originally interested in the number of bits that one needs to represent unit disk graphs on a small grid.
  Yet, their techniques readily provide an \ER-hardness reduction.

  \begin{theorem}
    It is \ER-complete to decide whether a graph is a unit disk intersection graph.
  \end{theorem}
  \begin{proof}
    It is easy to show \ER-membership.
    One way is to use the machine model of \ER:
    we first guess the coordinates of all disks and then verify the edges and non-edges by computing the distances between their centers.
    Another way is to build an \etr-formula.
    For each vertex $v \in V$ we introduce two variables $x_v$ and $y_v$ representing the coordinates of the disk center.
    For each edge $uv\in E$ we add the constraint
    $(x_v-x_u)^2 + (y_v-y_u)^2 \leq 2$, and for each non-edge $uv \notin E$ we add the constraint
    $(x_v-x_u)^2 + (y_v-y_u)^2 > 2$. 
    (We assume that unit disk means the radius to be equal to $1$.)

    To show \ER-hardness, we reduce from \stretchability.
    We begin with a high-level overview of the proof.
    For intuition, we first assume that \L is an arrangement of straight lines.
    This allows us to construct unit disks and their intersection graph $G$ in a systematic way.
    We then observe that $G$ is defined purely by the combinatorics of \L, and this forms the core of our construction.
    This mental picture---assuming that \L is stretchable---guides the intuition and clarifies how $G$ is defined.

    Let \L be a pseudoline arrangement. We will construct a graph from it.
    Actually, it makes sense to think of \L as a line arrangement to understand the underlying ideas.
    Given a line arrangement \L on $n$ lines $\ell_1,\ldots, \ell_n$, we can place one points in each cell arbitrarily. 
    We call those points $P$.
    Then, for each line $\ell$, there exist two disks $D^+$ and $D^-$ of equal radius such that $D^+$ contains all points above $\ell$, and $D^-$ contains all points below $\ell$.
    \begin{center}
      \includegraphics{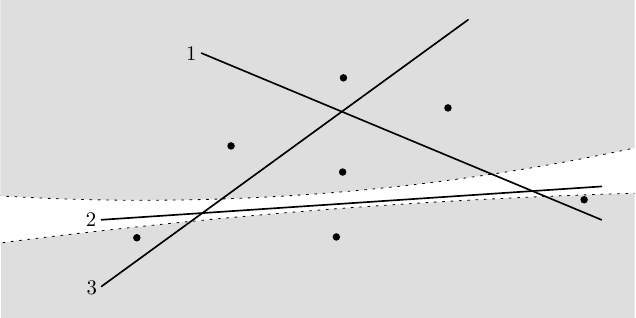}
      \captionof{figure}{Illustration of the large disks. 
      Note that if we blow up a disk, it becomes more and more similar to a half plane.}
    \end{center}
    We do this for every line and obtain a set
    $\mathcal{D}_1$ of $2n$ disks of equal radius $R$.
    (We can choose all radii to be equal as increasing the radius of a disk makes it more similar to a half plane.)
    In the next step, we halve the radius of each disk and keep its center to obtain a second set of disks $\mathcal{D}_2$.
    Next, we place for each point $p\in P$ a disk with radius $R/2$ and center $p$. 
    We define $G$ as the intersection graph of all those disks with radius $R/2$.
    We denote the vertices for the top and bottom disks by $v_i^+$ and $v_i^-$ respectively.
    \begin{center}
      \includegraphics[page = 2]{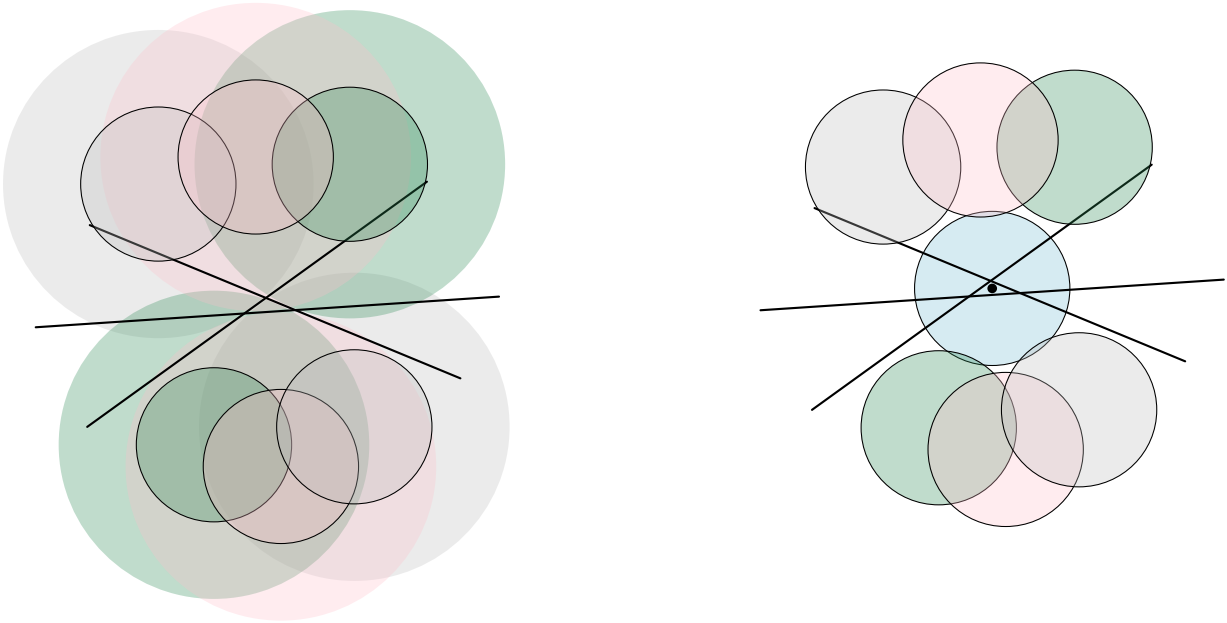}
      \captionof{figure}{Left: We halve the radius of each disk.
      Right: Then we add a disk with radius $R/2$ at each point of $P$. Here, we only indicate one disk in blue.}
      \label{fig:disk-reduction}
    \end{center}
    We will now argue the following points.
    \begin{itemize}
      \item The upper disks form a clique, and the lower disks form a clique.
      \item The graph $G$ can be defined purely from the pseudoline arrangement \L, without knowing if \L can be stretched.
      \item $G$ has at most $O(n^2)$ vertices  and thus the reduction is polynomial.
      \item If \L is stretchable, then $G$ is a unit disk intersection graph.
      \item Any realization of $G$ by unit disks defines a line arrangement equivalent to \L. 
    \end{itemize}

    For the first point, note that we can ``'squeeze'' the line arrangement without changing its combinatorics, as in \Cref{fig:disk-reduction-squeeze}.
    (In \Cref{fig:disk-reduction}, the upper disks do not form a clique to keep it readable.)
    \begin{center}
      \includegraphics[page = 4]{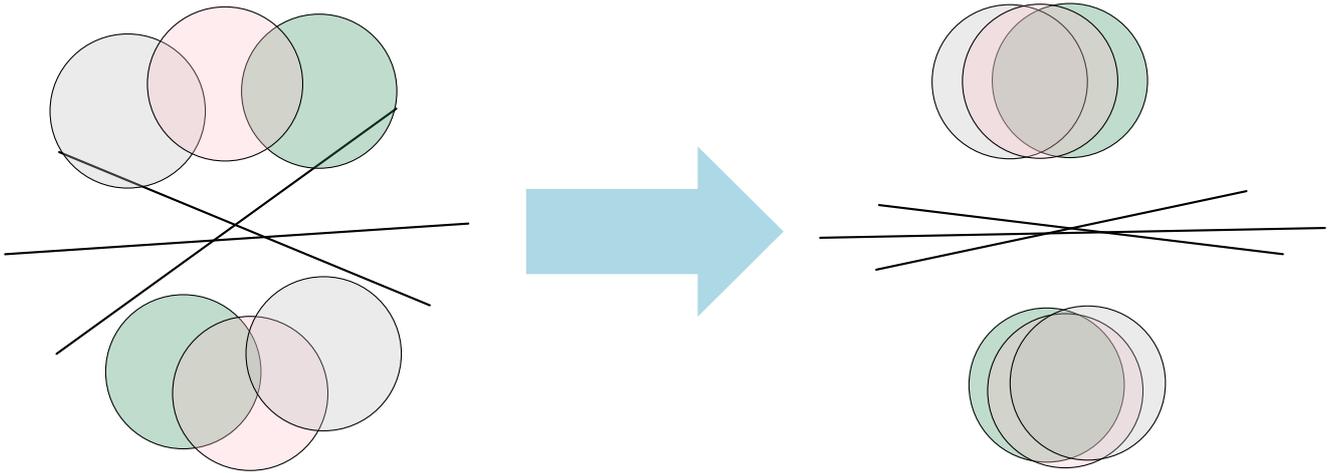}
      \captionof{figure}{When we squeeze the lines, we do not change the combinatorics, but the disks will all start to intersect.}
      \label{fig:disk-reduction-squeeze}
    \end{center}
    
    For the second point, we note that we know for each cell whether it is above or below a given pseudoline from the pseudoline arrangement and that is all we need to construct $G$.
    For the disks that come from each cell, we can compute the intersections merely by knowing whether they lie above or below a certain line.
    If the cell is above, we add edges to the $D^+$ disks; otherwise, we add edges to the $D^-$ disks.

    The third point is simple counting. 
    We construct two vertices per line and one vertex per cell.
    As there are only quadratically many cells, this gives the upper bound.

    The fourth point is clear as we used a stretched arrangement to construct a representation of  $G$ using unit disks.

    For the last point, given a line $\ell_i$, let $D^+_i$ and $D_i^-$ be the corresponding pair of disks, and let $k_i$ denote their bisector. 
    We claim that $k_1,\ldots,k_n$ forms a line arrangement equivalent to \L. 
    To see this, it is sufficient to show that it has combinatorially exactly the same set of cells.
    Let \L be our pseudoline arrangement and let $c$ be a cell in it.
    We denote by $D$ the disk with center $m$ inside $c$, and by $u$ the corresponding vertex in $G$.
    We will argue that $c$ is above $\ell_i$ if and only if $m$ is above $k_i$, for all $i$.
    If $c$ is above $\ell_i$, then $u$ is adjacent to $v^+_i$ and not to $v_i^-$. This implies that $m$ is closer to $D^+_i$ than to $D_i^-$, and hence that $m$ is above the bisector $k_i$.
    The other direction is identical.
    If $c$ is below $\ell_i$, then $u$ is adjacent to $v^-_i$ and not to $v_i^+$. This implies that $m$ is closer to $D^-_i$ than to $D_i^+$, and hence that $m$ is below the bisector $k_i$.

    This finishes the proof.
  \end{proof}

  While this proof is among the most elegant ones, it is a bit atypical; historically, it was not the first such reduction.
  For most interesting sets of geometric objects, we now know that the recognition problem is \ER-complete~\cite{SchaeferCardinalMiltzow2024Compendium}.
  Often these proofs are straightforward, but sometimes they require substantial scaffolding to encode the pseudoline arrangement properly.
  All of them are based on encoding pseudoline arrangements.
  To get a better idea of these scaffolding techniques, we prove another theorem.

\paragraph{Optimal Curve Straightening.}
We briefly motivate the \textsc{Optimal Curve Straightening Problem} and give a clean definition.
Suppose we are given a closed curve $\gamma$ in the plane, possibly with self-intersections.
We can naturally describe $\gamma$ by the $4$-regular plane graph induced by its self-intersection pattern.
In many settings (visualization, compression, or geometric algorithms) we want a simpler curve that \emph{represents} $\gamma$ and uses as few segments as possible.
Specifically, we do not want to change its topological behavior, that is, the plane graph arising from the self-intersection pattern should be the same.
See \Cref{fig:curve-straightening-definition} for an illustration.
\begin{center}
    \includegraphics[page = 1]{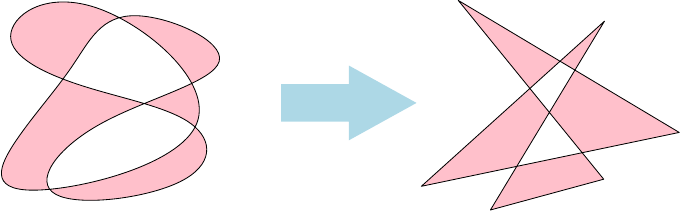}
    \captionof{figure}{A closed curve in the plane and a representation with six segments.}
    \label{fig:curve-straightening-definition}
  \end{center}
In the \textsc{Optimal Curve Straightening Problem}, we are given a closed curve $\gamma$ and an integer $k$ and we are asked if there is a way to represent $\gamma$ with at most $k$ segments.
We show the following theorem due to Erickson~\cite{Erickson2019OptimalCurveStraightening}.

\begin{theorem}
  \label{thm:Curve-Straightening}
  The \textsc{Optimal Curve Straightening Problem} is \ER-complete.
\end{theorem}

\subparagraph{\ER-Membership.}
  While it is easy to describe $k$ segments using an \etr-formula it is difficult to encode that the planar graph derived from its self-intersection pattern is the same.
  Erickson wrote a five page proof describing an \etr-formula that does the job. 
  This complicated \ER-membership proof was one of the motivations to find a machine model.
  Here, we only give the \ER-membership argument using the machine model.
  Our witnesses are the coordinates of the $k$ segments, and the verification algorithm computes the intersection pattern and checks that the two planar graphs are identical.

\subparagraph{\ER-Hardness.}
We reduce from \stretchability. 
For this purpose, assume that \L is an arrangement of $n$ pseudolines.
We construct a curve $\gamma$ such that $\gamma$ can be represented with $4(n+1)$ segments if and only if \L is stretchable.

To begin the construction, we assume that $n$ is odd. 
Otherwise, add one more pseudoline that intersects all pseudolines to the far left. This does not affect whether \L is stretchable.

The next step is to take a big rectangle that contains all intersections of the pseudolines.
Then, we add one extra line on top that does not intersect any other pseudoline.
Finally, we stitch the ends together with loops as displayed in \Cref{fig:curve-straightening-reduction}.
We set $k = 4(n+1)$ and this finishes the construction of $\gamma$.

  \begin{center}
    \includegraphics[page = 2]{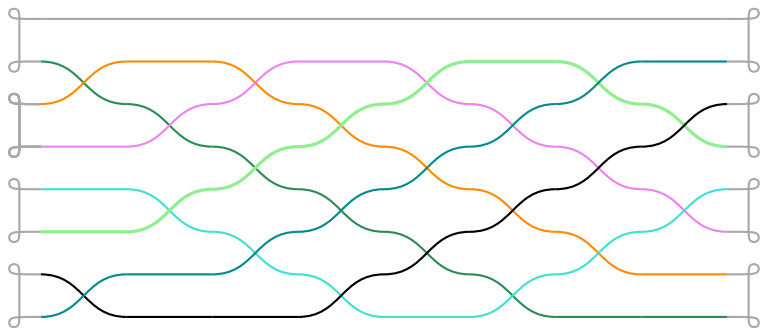}
    \captionof{figure}{Constructing the curve $\gamma$ from the pseudoline arrangement \L.}
    \label{fig:curve-straightening-reduction}
  \end{center}

It is clear that we can construct $\gamma$ in polynomial time.
Next, we have to show that 
\L is stretchable
if and only if 
$\gamma$ can be represented with $4(n+1)$ line segments .
We start with the first implication.
For this, let $L$ be the set of (stretched) straight lines representing \L.
We again find a big rectangle containing all intersections of $L$ and restrict all the lines to the rectangle.
We add a line on top of all other lines, not intersecting any of them above the rectangle.
Finally, we add all the little curls at the left and right ends of each pair of adjacent line segments, as in \Cref{fig:curve-straightening-straight-curve}.

\begin{center}
    \includegraphics[page = 4]{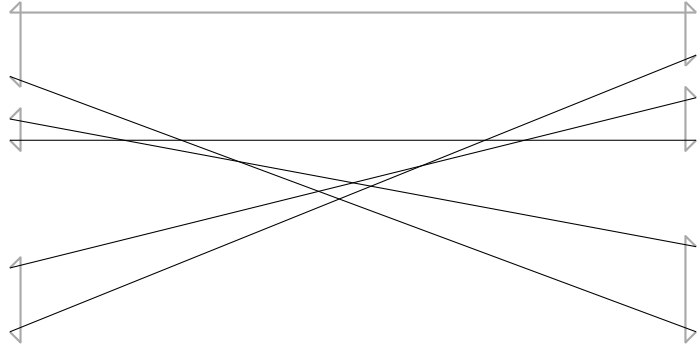}
    \captionof{figure}{Constructing a curve from a stretched arrangement.}
    \label{fig:curve-straightening-straight-curve}
  \end{center}

To prove the reverse direction, let $D$ be a straight line drawing of $\gamma$ with $4(n+1)$ vertices.
  \begin{center}
    \includegraphics[page = 3]{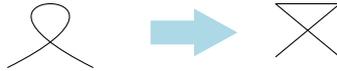}
    \captionof{figure}{Any representation of the curl needs at least $2$ vertices.}
    \label{fig:curve-straightening-curl}
  \end{center}

Note that every little curl at the end needs at least two vertices to be represented. We have $2(n+1)$ of those curls, and thus any polygonal line representing $\gamma$ needs at least $4(n+1)$ vertices, and thus at least that many line segments.
See \Cref{fig:curve-straightening-curl}.
Now, if we consider the parts of the drawing $D$ that does correspond to the original pseudolines, each such pseudoline must be represented by a straight line segment, by a simple counting argument.
(If there would be a vertex on one of the pseudolines, we had not enough vertices left for the curls.)
When we extend all of those line segments ot infinite lines, we receive a line arrangement $L$ equivalent to the original pseudoline arrangement \L.

This finishes the proof of \Cref{thm:Curve-Straightening}.

\subsection{Exercises}
These exercises are meant to review and deepen the content of this section.
\begin{enumerate}[noitemsep]
  \item Let $P$ be a set of $n$ points in the plane.
  \begin{enumerate}[noitemsep]
    \item Find an algorithm that returns a line $\ell$ such that half the points are on the left and half are on the right of the line. You are allowed to use the coordinates of the points.
    \item Let $\ell$ be a line defined by two points. We say $\ell$ is a halving line if at least $n/2$ points are to the left and to the right of $\ell$.
    This time, find a halving line using only the order type of the point set. (Can you do it in linear time?)
    \item Show that the convex hull of a point set is completely determined by its order type.
  \end{enumerate}
  \item Let $p,q$ be two polynomials with integer coefficients.
  Can you express $(p=0 \land q=0)$ as a single equation?
  \item In \textsc{4-Feasibility}, we are given a polynomial of degree four with integer coefficients. We have to decide whether there are real numbers that make this polynomial evaluate to zero.
  Show that \textsc{4-Feasibility} is \ER-complete.
  \item Let $s_1,\ldots,s_n$ be segments in the upper half plane ($H = \{(x,y)\in \R^2 : y \geq 0\}$.), with one endpoint on the $x$-axis.
  We call those \textit{grounded segments} and the corresponding intersection graph a grounded segment graph.
  In the same manner, we can define upward ray graphs as the intersection graphs formed by finite collections of upward rays. 
  A ray is an upward ray, if its direction vectror $v = (x,y)$ has positive $y$-coordinate.
  Show that the upward ray graphs and grounded segment graphs form the same set of graphs.
  (Hint: Find the right projective transformation.)
  \item Let $s= ab$ be a segment from the point $a$ to the point $b$.
  Can you describe the dual of $s$?
  \item When we rotate a point set very slowly, how does the dual line arrangement of the point set change combinatorially?
  (Hint: Consider the left and right most intersection of two lines in the dual.)
  \item Find a graph $G$ that cannot be represented as unit disk intersection graph.
  \item In the proof of \Cref{thm:Curve-Straightening} we assumed that $n$ is odd. What goes wrong if $n$ is even?
\end{enumerate}

\paragraph{Open Pool Exam Questions.}
These questions capture the learning goals of the section.
 \begin{enumerate}[noitemsep]
        \item Give a definition of the algorithmic problem \OrderTypeProblem and the \partialOrderTypeProblem.
        \item Define \textsc{ETR-AM} and show that it is \ER-complete.
        \item Describe how to encode the value of a variable using a point on a line.
        \item Describe the van Staudt construction for addition and multiplication assuming that you can enforce parallelism.
        Show that they indeed enforce addition and multiplication using the interpretation from the previous assignment.
        \item State the key properties of the projective plane that we use.
        \item Describe the van Staudt construction for addition without using parallelism.
        \item Describe point-line duality.
        \item Show the above-below property of point-line duality.
        \item Define the cyclic order of a point with respect to a point set. Describe what the cyclic order corresponds to in the dual.
        \item Define geometric intersection graphs and unit disk graphs in particular.
        \item Show that deciding whether a given graph is a unit disk intersection graph is \ER-complete.
        \item Show that \textsc{Optimal Curve Straightening} problem is \ER-complete.
    \end{enumerate}

\newpage
\section{Inversion and Compactness.}
While \stretchability served as a central problem for proving \ER-hardness, a number of problems required new techniques to show \ER-hardness.
The new central problem is called \etrinv. 
It has two main properties. 
First, the domain is a small compact interval, for example $[1/2,2]$.
Second, instead of multiplication, we only allow inversion constraints of the form $x\cdot y = 1$. 
In this section, we define \etrinv and explain two techniques for proving \ER-hardness. 
We will not combine these two techniques, since doing so is tedious and offers little additional insight.
Thereafter, we discuss the historical context, linked to the art gallery problem, and the impact this had.
Finally, we apply these ideas to show that training neural networks is \ER-complete---via a surprisingly short and elegant proof.

\subsection{\etrinv}

The algorithmic problem \etrinv asks whether there exist real variables $x_1,\ldots,x_n$ in the compact range $[1/2,2]$ such that a set of constraints is satisfied.
Each constraint has the form $x+y=z$ or $x\cdot y = 1$, with $x,y,z\in\{x_1,\ldots,x_n\}$.

\paragraph{Inversion.}
We first show how we can restrict ourselves to inversion constraints.
For that purpose, we define the problem \textsc{UN-INV}, an abbreviation for \emph{unbounded inversion}, referring to the fact that the variables range over all real numbers.
\textsc{UN-INV} is defined exactly as \etrinv, with the small difference that all variables are allowed to take any real value.
Additionally, we need the constraint $x = 1$.
We show the following lemma.

\begin{lemma}
  \textsc{UN-INV} is \ER-complete.
\end{lemma}
\begin{proof}
  We start with \ER-membership. We simply note that \textsc{UN-INV} is a special case of \etr and thus lies in \ER by definition.

  To show \ER-hardness, we reduce from \textsc{ETR-AM}.
  Let $\varphi$ be an instance of \textsc{ETR-AM}. We show how to transform it in polynomial time into an instance of \textsc{UN-INV}.
  We do so in two steps. First, we replace multiplication by squaring ($x^2 = y$) and then we replace squaring by inversion.

  Let $x\cdot y = z$ be a multiplication constraint of $\varphi$.
  First note that
  \[(x+y)^2 - x^2 - y^2 = 2x\cdot y.\]
  To use this insight, we need to define intermediate variables as follows.

  \begin{align*}
    A &= x+y\\
    B &= A^2\\
    C &= x^2\\
    D &= y^2\\
    E &= C + D\\
    B &= F + E\\
    F &= z + z
  \end{align*}

A short calculation shows that $z = x\cdot y$ and that we have used only addition and squaring constraints.
All new variables are existentially quantified.
We perform this replacement for every multiplication constraint.

The next step is to replace every squaring constraint ($x^2 = y$), by an inversion constraint.
To do so, we observe that \[\frac{1}{x} - \frac{1}{x+1} = \frac{1}{x(x+1)}.\]
And thus it follows
\[\frac{1}{\frac{1}{x} - \frac{1}{x+1}} - x = x^2.\]
We use this insight to construct appropriate intermediate variables.

\begin{align*}
    A \cdot x &= 1 \\
    B &= x+1\\
    C \cdot B &= 1\\
    D + C &= A\\
    E\cdot D &= 1\\
    y &= E + x
  \end{align*}
A short calculation shows that this enforces $y = x^2$ and that we have used only inversion and addition constraints.
In the final step, we replace all squaring constraints in this way in linear time.
This finishes the proof.
\end{proof}

Finally, we mention that $x\cdot x = 1$ implies that $x = 1$ or $x = -1$.
Since in \etrinv all variables lie in the range $[1/2,2]$, we have $x=1$, and the constraint $x=1$ is therefore not needed.

\paragraph{Compactification.}
  In many applications, it is easy to encode a variable as a geometric object.
  However, the object necessarily lies in a compact domain, that is, bounded and closed.
  Thus, for a reduction to be correct, the domain of the variable we want to encode must be compact as well.
  We first show how to encode a bounded variable and then, using a different technique, how to encode a compact variable.
  For the purpose of illustration, we define \textsc{Bounded-ETR} in the same way as \etr, with the difference that every variable is bounded to the open interval $(-1,1)$.
  We show the following lemma.
  \begin{lemma}
    \textsc{Bounded-ETR} is \ER-complete.
  \end{lemma}
  \begin{proof}
    \textsc{Bounded-ETR} is a special case of \etr and this implies \ER-membership.

    To show \ER-hardness, we reduce from \textsc{ETR-AM}.
    Consider the following mapping $f : (-1,1) \rightarrow \R$, defined by
    \[f(x) = \frac{x}{(x-1)(x+1)}.\]
    Note that $f$ is a continuous bijection.
    Now, let $\varphi$ be an \textsc{ETR-AM} instance.
    We replace each occurrence of each variable in a constraint by $f(x)$ and then multiply out all denominators.
    For example, $x+y = z$ becomes
    \[\frac{x}{(x-1)(x+1)} + \frac{y}{(y-1)(y+1)} = \frac{z}{(z-1)(z+1)},\]
    which becomes
    \[
      x (y-1)(y+1)(z-1)(z+1)
      + y (x-1)(x+1)(z-1)(z+1)
      = z (x-1)(x+1)(y-1)(y+1).
    \]
    Clearly, this creates an instance of \textsc{Bounded-ETR} $\psi$.
    Now, if $(x_1,\ldots,x_n) \in (-1,1)^n$ is a solution to $\psi$, then
    $(f(x_1),\ldots,f(x_n))$ is a solution to $\varphi$.
    Conversely, if $(x_1,\ldots,x_n) \in \R^n$ is a solution to $\varphi$, then
    $(f^{-1}(x_1),\ldots,f^{-1}(x_n))$ is a solution to $\psi$.
  \end{proof}

  It is interesting that the previous lemma can be shown in such an elementary way.
  While in some applications open bounded domains are fine, this is typically not the case as we will see.
  Compactification needs much more sophisticated tools.
  One way to see this is that we used the existence of a continuous bijection between an open interval and the reals. There is no continuous bijection between a closed interval and the reals as they are topologically different spaces.
  The next tool we introduce is the so-called ``Ball Theorem''.

\begin{theorem}[Ball Theorem]
  \label{thm:ball}
  Every semi-algebraic set in $R^n$ of complexity at most $L \geq 4$ 
  contains a point of distance at most $2^{L^{8n}}$ from the origin.  
\end{theorem}
Here, the complexity $L$ of a semialgebraic set $S$ is the length of the shortest formula $\varphi$ such that \[S = \{x\in \R^n : \varphi(x)\}.\]

We use a version distilled by Schaefer and {\v{S}}tefankovi{\v{c}}~\cite{Schaefer2017FixedPointsNashETR}.
The original research was done in algebraic geometry, and the first version of a ball theorem can be traced back to 1984~\cite{Vorobev1984EstimatesRealRoots}.
Here we focus on how to apply the theorem.
For concreteness, we define \textsc{Compact-ETR} in the same way as \etr, but with every variable restricted to the compact domain $[-1,1]$.
We prove the following lemma.

\begin{lemma}
  \textsc{Compact-ETR} is \ER-complete.
\end{lemma}
\begin{proof}
  Membership in \ER is immediate, as \textsc{Compact-ETR} is a special case of \etr.

  To show hardness, let $\exists x\in \R^n : \varphi(x)$ be an instance of \textsc{ETR-AM} with description length $L$.
  We assume that $L\geq 4$, since our reduction only needs to work for larger instances.
  
  We now construct an equivalent instance
  $\exists x\in \R^m :\psi(x)$ of \textsc{Compact-ETR}.
  Let $M \geq 2^{L^{8n}}$ as in \Cref{thm:ball}.
  We immediately obtain that
  $\exists x\in \R^n : \varphi(x)$ is equivalent to 
  $\exists x\in [-M,M]^n : \varphi(x)$.
  To see this, let $S = \{x\in \R^n : \varphi(x)\}$.
  By \Cref{thm:ball}, if $S$ is non-empty, then it contains a point $x$ at distance at most $M$ from the origin,
  and hence in the cube $[-M,M]^n$.
  Therefore, we can restrict to this compact range.
  All that remains is to shrink the range.
  We use a similar trick as before, defining a continuous bijection $f : [-1,1] \rightarrow [-M,M]$ by $f(x) = M \cdot x$.
  To use this bijection, we do need to be able to construct $M$.
  And $M$ itself is not in the range $[-1,1]$.
  To get around this, we construct $\varepsilon = 1/M$, and our function $f$ can be encoded as
  $f(x) = x/\varepsilon$. 
  In fact, we will construct $\varepsilon$ such that $\varepsilon^{-1} = M = 2^{2^k} \geq 2^{L^{8n}}$ using a formula of size $O(k) = O(n \log L)$.
  We are using repeated squaring for this.
  Let $x_1 = 1/2$ and \[x_{i+1} = x_i^2.\]
  It is easy to see that $x_k = 2^{-2^k} = \varepsilon$.
  
  Now the formula $\Psi = \exists x\in [-1,1]^m \psi(x)$ has the following form.
  The first part constructs a variable that must hold the value $\varepsilon$ as above.
  The second part of $\psi$ is exactly the same as $\varphi$, with every occurrence of every variable $x$ replaced by $x/\varepsilon$.
  We multiply this out everywhere to avoid divisions.
  For example 
  $x\cdot y = z$ becomes
  \[\frac{x}{\varepsilon}\frac{y}{\varepsilon} = \frac{z}{\varepsilon},\]
  which is equivalent to 
  \[xy = z \varepsilon.\]
\end{proof}

\subsection{The \artgalleryProblem.}
The \artgalleryProblem is a classical problem in computational geometry, motivated by a simple and intuitive scenario.
Imagine an art gallery with a complicated floor plan consisting of walls, corners, and narrow corridors.
A guard placed inside the gallery can observe its surroundings, but walls block visibility.
The central question is: how many guards are necessary, and where should they be placed, so that every point of the gallery is visible to at least one guard?

Despite its simple formulation, the problem reveals deep connections between geometry, combinatorics, and algorithms.
It gave rise to fundamental results such as Chvátal’s theorem~\cite{Chvatal1975, Fisk1978}, which states that any simple polygon with \(n\) vertices can be guarded by at most \(\lfloor n/3 \rfloor\) guards.
The art gallery problem has since become a cornerstone for the study of visibility.

Let \(P \subset \mathbb{R}^2\) be a \emph{simple polygon}.
For two points \(x,y \in P\), we say that \(x\) \emph{sees} \(y\) if the closed line segment $\overline{xy} \subseteq P$.
A set \(G \subseteq P\) is called a \emph{guarding set} for \(P\) if for every point \(p \in P\) there exists a guard \(g \in G\) such that $g$ sees $p$.
The \emph{art gallery problem} asks to determine, for a given polygon \(P\), a guarding set \(G\) of minimum cardinality.

\begin{center}
\includegraphics[scale = 0.3]{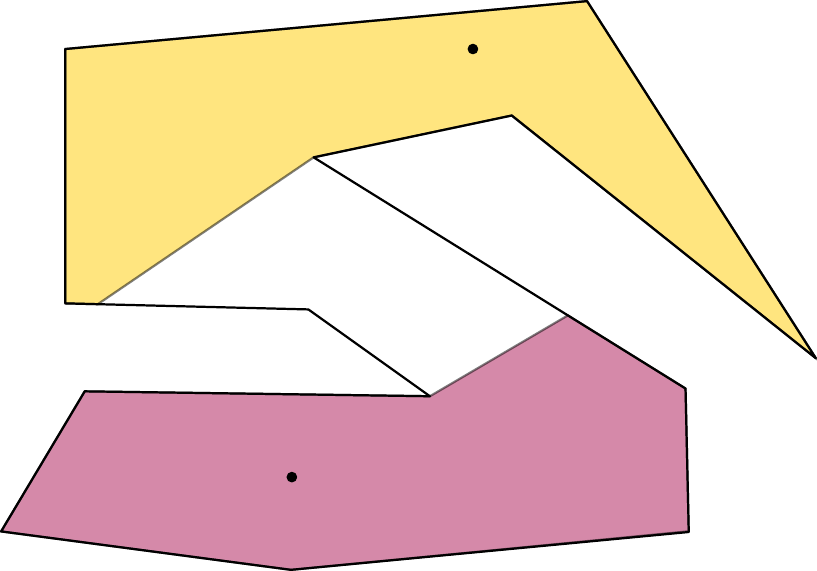}
\end{center}
In the example, we can see a polygon together with two guards that are guarding part of the entire polygon.

In 1984 Aggarwal showed that the \artgalleryProblem is \NP-hard~\cite{aggarwal1984artgallery}.
Later, many more hardness proofs for special cases and variants followed, but only very few researchers mentioned the problem of \NP-membership~\cite{orourke1987artgallery}. 
In fact, many researchers that I spoke to during my time as a PhD student told me that there is probably an easy argument that the \artgalleryProblem is in \NP, but they happen to have it not ready at the moment.
But maybe it would be an instructive exercise for me.
This intuition from those at the time senior researchers came likely from the fact that in most instances you can wiggle the solution a little bit.
Then you find a rational solution with low coordinate complexity and this will be your NP-membership witness.
There were only few cases known where the guard positions are fixed and in those cases one could easily compute the coordinates and thus they can also be described with few bits.
While this intuition is appealing, I did not manage to turn it into a rigorous proof.
And there is a good reason for that.
In 2017 Abrahamsen, Adamaszek and Miltzow~\cite{abrahamsen2017irrational} found a polygon that required irrational coordinates.
\begin{center}
    \includegraphics{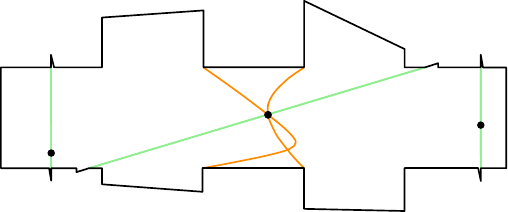}
\end{center}
The polygon is surprisingly simple and it boils down to some simple calculations.
As a matter of fact the same authors showed later that the following theorem~\cite{abrahamsen2022artgallery}.
We recommend to read the more accessible proof by Stade~\cite{Stade2025PointBoundaryArtGallery}.

\begin{theorem}
  The \artgalleryProblem is \ER-complete.  
\end{theorem}

This proof is very long and at times very technical.
There are a few very simple key ideas here that are instructive and became useful also for later reductions and other problems.
Most prominently the introduction of \etrinv.
We give here two small parts of the proof that reflect some of the core ideas, namely, on how to encode variables into the \artgalleryProblem and how to encode the inversion constraint.
Similar ideas have been used in many reductions that followed.

\paragraph{Variable Representation.}
  We represent a variable using a guard, and we enforce that the guard can only be placed on a single compact segment, as displayed in \Cref{fig:variable-guard-segment}. 

    \begin{center}
      \includegraphics{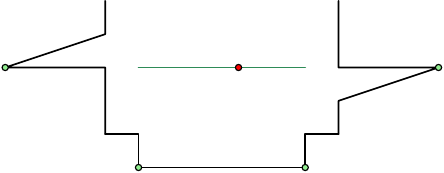}
      \captionof{figure}{Any guard that wants to see the four green dots simultaneously must be placed on the green segment.}
      \label{fig:variable-guard-segment}
    \end{center}

It is clear that this variable segment must have bounded length, since it lies inside a polygon that is itself bounded.

It is necessary that this segment be compact for the following reason.
We enforce the guard $g$ to see a set of points $p_1,\ldots,p_k$.
This means that $g$ must be in the intersection of the visibility regions, i.e., 
$g \in \Vis(p_1)\cap \ldots \cap \Vis(p_k)$.
Here $\Vis(p)$ is the region seen by point $p$ and this region is by definition of visibility compact.
Thus, the region to which $g$ is restricted is compact as well.
We stress this point, because this is the justification for using the ball theorem and not a simpler way to restrict the variable range.

\paragraph{Inversion Representation.}
  We construct the constraint $x\cdot y \geq 1$.
  The construction of the constraint 
  $x\cdot y \leq 1$ is technically more involved, but essentially the same. 
  The construction is depicted in the figure below.
  It consists of two guard segments (we omitted the nooks that enforce the guards to lie on the segments, for clarity), a critical wall to the right, and two pivot points that block the visibility of the green and red guards.

  \begin{center}
    \includegraphics{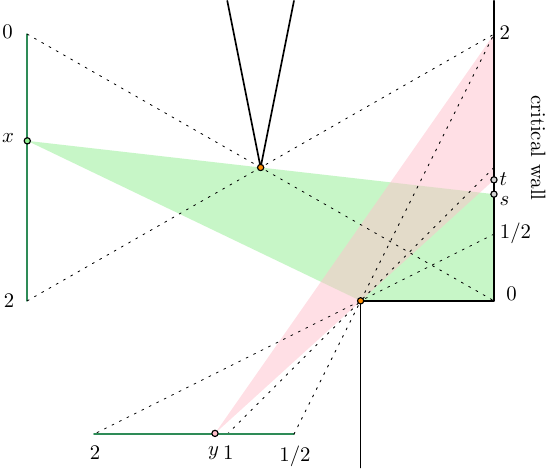}
    \captionof{figure}{The two guards can see the wall to the right together if and only if they together satisfy the inequality $x\cdot y \geq 1$.}
      \label{fig:art-inversion}
  \end{center}

  We show the following lemma.
  \begin{lemma}
    Assume that no other guard interferes and that the two guards in \Cref{fig:art-inversion} are on their respective guard segments.
    Let $x,y$ be the values that their position represents.
    Then they can see together the critical wall on the right if and only if $x\cdot y \geq 1$. 
  \end{lemma}
  \begin{proof}
    For convenience, we mark the critical wall also with values that help us to explain the proof.
    
    Let's first consider the green guard representing $x$. 
    We denote by $s$ the value on the critical segment of the most upper point that is seen by the green guard.
    Similarly, we define $t$ to be the value of the lowest point seen by the red guard.
    It is clear that the two guards see the entire critical segment if and only if $s\geq t$.

    Since the green guard segment is parallel to the critical wall and the pivot point lies exactly midway between the two segments, we have $x = s$.

    It remains to show that $y = 1/t$.
    To see this, consider two similar triangles, $\Delta_1$ and $\Delta_2$.
    \begin{center}
    \includegraphics[page = 2]{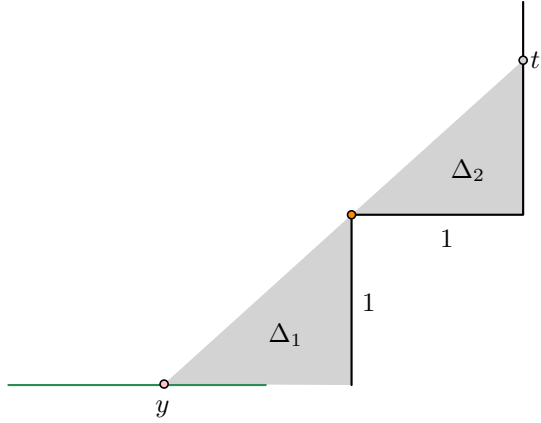}
    \captionof{figure}{The two triangles $\Delta_1$ and $\Delta_2$ are similar. }
      \label{fig:triangle-inversion}
    \end{center}
    Note that $\Delta_1$ has height $1$ and width $y$.
    Similarly, $\Delta_2$ has height $t$ and width $1$.
    Thus, we have $\frac{1}{y} = \frac{t}{1}$, which shows the claim and finishes the proof.
  \end{proof}

  Representing inversion is typically much easier than representing multiplication, because we have a non-linear interaction between three objects that we have to control precisely. 
  In case of inversion, we only have to control the interaction between two objects. 

  Interestingly, neither known \ER-hardness proof of the \artgalleryProblem uses this gadget to enforce inversion.
  While this gadget is much easier than the gadgets that were used, it has two downsides.
  One is that the two guard segments are not parallel, and copying values from a horizontal guard segment to a vertical guard segment is not easy. 
  Secondly, it is hard to connect this gadget to the rest of the construction.
  Still, this gadget shows that it is possible to encode inversion, and it remains an intellectual challenge to find a gadget that works within an \ER-hardness reduction.

\vfill

\paragraph{Acknowledgements}
We thank Michael Dobbins for helpful discussions on pseudolines and polytopes.
We thank Sebastiaan Jans for numerous comments on typos in Section 1.
We thank Rens de Wit for pointing out a mistake in the calculation of the multiplication gadget in the partial order type realizability proof.
We thank Taylan Batman for suggestions that improved the clarity of the greatest common divisor algorithm, the \ER-completeness  proof of recognizing unit disk intersection graphs and the \ER-hardness of optimal curve straightening.
We thank Dinand Blom for pointing out possible improvements to the presentation of the \ER-hardness proof for recognizing unit disk intersection graphs.
We thank Daan van Dam for pointing out a typo for the linear search algorithm on the word RAM.

\newpage
\bibliographystyle{abbrv}
\bibliography{library}

\end{document}